\documentclass[12pt]{amsart}

\usepackage[matrix,arrow,curve,frame]{xy}
\usepackage{amsmath,amsthm,amssymb,enumerate, euscript}
\usepackage{latexsym,amscd, xcolor, tikz, tikz-cd, comment, fullpage}
\usepackage[backref]{hyperref}
\usepackage[all]{xypic}
\newtheorem{theorem}{Theorem}[section]
\newtheorem{lemma}[theorem]{Lemma}

\newtheorem{proposition}[theorem]{Proposition}

\theoremstyle{definition}
\newtheorem{definition}[theorem]{Definition}

\newtheorem{remark}[theorem]{Remark}

\numberwithin{equation}{section}
\newcommand {\be}{\begin{equation}}
\newcommand {\ee}{\end{equation}}

\newcommand{\TT}{\mathbb{T}}
\newcommand{\ZZ}{\mathbb{Z}}
\newcommand{\RR}{\mathbb{R}}
\newcommand{\Q}{\mathbb{Q}}
\newcommand\bbT{\mathbb T}
\newcommand\GG{\mathbb{G}}
\newcommand{\hol}{\mathrm{hol}}
\newcommand{\sU}{{\sf U}}

\renewcommand{\cL}{\mathcal{L}}
\newcommand{\cG}{\mathcal{G}}
\newcommand{\Ca}{\mathcal{C}}
\newcommand{\pic}{\mathrm{Pic}^{\nabla}}
\newcommand{\Z}{\mathbb{Z}}
\newcommand{\Caf}{\mathcal{C}^{flat}}
\newcommand{\Asigma}{\mathcal{A}_{\Sigma, U(1)}}
\newcommand{\dvol}{\mathrm{dvol}}
\newcommand{\set}[1]{\left\{\,#1\,\right\}}
\providecommand{\Tors}{\operatorname{Tors}}

\DeclareMathOperator{\coker}{coker}
\DeclareMathOperator{\rank}{rank}
\DeclareMathOperator{\im}{im}

\begin{document}
 
\title[Moduli Spaces of Connections and B-fields from T-duality with $H$-flux]
{Moduli Spaces of Flat Mixed Fields and T-duality}
 
\author{Fei Han}
\address{Department of Mathematics,
National University of Singapore, Singapore 119076}
\email{mathanf@nus.edu.sg}
 
\author{Pedram Hekmati}
\address{University of Auckland Department of Mathematics 38 Princes Street, Auckland 1010, New Zealand}
\curraddr{}
\email{p.hekmati@auckland.ac.nz}

\author{Tsuyoshi Kato}
\address{Department of Mathematics, Graduate School of Science, Kyoto University, Kyoto 606-8502, Japan}
\email{tkato@math.kyoto-u.ac.jp}

\author{Varghese Mathai}
\address{School of Mathematical Sciences,
Adelaide University, Adelaide 5005, Australia}
\email{mathai.varghese@adelaide.edu.au}
 
\date{}

\maketitle

\begin{abstract} 
We study the geometry of mixed fields, consisting of a connection and a
$B$-field on a principal circle bundle over a Riemann surface, from the perspective of gauge theory
and T-duality. Motivated by the foundational work of Atiyah--Bott and Segal,
we introduce a twisted Yang--Mills functional whose critical locus, in the
flat case, is governed by the simultaneous vanishing of the curvature and the
$H$-flux. We show that the gauge group is a semi-direct product of abelian groups parametrised by an integer  $\lambda$.
The moduli spaces are constructed by presymplectic reduction and shown to be
 Heisenberg contact manifolds for $\lambda \neq 0$, whose topology we characterise completely. 
We show that T-duality preserves the twisted Yang--Mills functional and acts on the configuration space of flat mixed fields. We
identify the subgroups of gauge transformations that are compatible with the T-duality map and describe the induced action on the 
singular quotient of T-dualizable flat mixed fields. Precisely at $\lambda =1$  does  T-duality descend to an involutive contactomorphism of the moduli space.
\end{abstract}


\section{Introduction}
\subsection{Background and motivation}
T-duality in string theory can be understood as a transformation acting on
the worldsheet fields of a two-dimensional nonlinear sigma model~\cite{HLSUZ}.
In the presence of supersymmetry, this duality has been studied in
\cite{LRRUZ, KL}. When the target manifold admits an abelian isometry which
preserves the worldsheet action, the construction becomes particularly
transparent: one obtains not only an explicit transformation of the
worldsheet fields, but also the Buscher rules~\cite{Buscher}, which determine
the dual background. The procedure is symmetric in nature, in the sense that
starting from either the original theory or its dual, the same gauging
argument reproduces the other side of the correspondence~\cite{Alvarez2}.

\medskip

From a geometric point of view, T-duality for principal circle bundles with
flux was formulated in~\cite{BEM04a, BEM04b}.
Related perspectives involving twisted $K$-theory, bundle
gerbes and D-branes may be found in~\cite{BM00,BCMMS,MS}. 
Topological and noncommutative formulations were developed in
\cite{BunkeSchick,MR05}, while the relation with generalized geometry is
discussed in~\cite{CG}. Let $Z$ be a principal
$\mathbb{T}$-bundle over a manifold $X$
\[
\begin{CD}
\mathbb{T} @>>> Z \\
&& @V\pi VV \\
&& X,
\end{CD}
\]
equipped with an $H$-flux, that is, a closed $3$-form
$H\in \Omega^3(Z)$ with integral periods (we suppress the factor
$\frac{1}{2\pi i}$ for simplicity). Let $\{U_\alpha\}$ be a good cover of $X$,
and let $A$ be a connection $1$-form on $Z$. We choose local $2$-forms
$B=\{B_\alpha\}$ such that
\[
H|_{\pi^{-1}(U_\alpha)} = dB_\alpha,
\qquad B_\alpha \in \Omega^2(\pi^{-1}(U_\alpha)).
\]
Associated to this data $(Z,A,B,H)$, there is a canonical T-dual quadruple
$(\widehat Z,\widehat A,\widehat B,\widehat H)$, where
\[
\begin{CD}
\widehat{\mathbb{T}} @>>> \widehat Z \\
&& @V\widehat{\pi}VV \\
&& X
\end{CD}
\]
is a principal $\widehat{\mathbb{T}}$-bundle over $X$, equipped with the dual
connection $\widehat A$, dual $B$-fields $\{\widehat B_\alpha\}$, and dual
flux $\widehat H$.
The dual pair is characterized by the relations
\[
\pi_*(H)=F_{\widehat A}, 
\qquad 
\widehat\pi_*(\widehat H)=F_A,
\]
where $F_A=dA$ and $F_{\widehat A}=d\widehat A$ denote the curvatures.
In this way, one obtains a T-duality transformation
\[
\mathcal{T}\colon (Z,A,B)\longmapsto (\widehat Z,\widehat A,\widehat B).
\]
A fundamental consequence of this correspondence is that it induces
isomorphisms in twisted cohomology and twisted $K$-theory between the dual
spaces~\cite{BEM04a, BEM04b}.
A more detailed review will be given in
Section~\ref{rev}.

\medskip

From the viewpoint of gauge theory, one is naturally led to consider
abelian gauge fields on $Z$ together with $B$-fields encoding
gerbe-theoretic data.
The work of Atiyah and Bott~\cite{AB} revealed that the space of connections carries
a natural symplectic form, and the gauge group acts in a Hamiltonian
fashion.
This fits into the broader framework of Hamiltonian group
actions, moment maps and symplectic reduction developed in
\cite{MW,GS,DH82,DH83}.

The present paper extends this Atiyah--Bott framework by incorporating
$B$-fields into the configuration space. We therefore consider
\be \label{conf}
\mathcal{C}=\{(A,B)\},
\ee
consisting of a connection $A$ on $Z$ together with a $B$-field as in the T-duality picture. We
shall refer to such pairs as \emph{mixed fields}.

\medskip

Fix a Riemannian metric $h$ on $X$. The connection $A$
determines a $\bbT$-invariant metric on $Z$ by
\[
h_A:=\pi^*h+A\odot A,
\]
where $A\odot A$ denotes the symmetric product.
Motivated by the seminal work of Segal~\cite{Segal} and Brylinski \cite{Bry98}, we introduce the
\emph{twisted Yang--Mills functional}
\[
\mathrm{YM}(A,B)
:=
\int_X |F_A|^2\,d\operatorname{vol}_h
+
\int_Z H \wedge *_{h_A} H,
\]
which extends the classical Yang--Mills functional by the additional
term involving the $H$-flux. The second term measures the $L^2$-energy of the
$H$-flux with respect to the metric $h_A$ determined
by the connection $A$. 

\medskip

We will see that when $X=\Sigma$ is a closed Riemannian surface of genus $g$ and
$\langle c_1(Z),[\Sigma]\rangle=0$, the minima of this functional
are precisely those mixed fields satisfying
\[
F_A=0,
\qquad 
H=0,
\]
which we call \emph{flat mixed fields}. Moreover, the two-dimensionality of the base   
allows the gauge group of the circle bundle and the symmetry group of the 
$B$-field to couple, so rather than acting independently, they combine into a semi-direct 
product indexed by a {\em level} $\lambda\in \Z$.
The resulting moduli spaces of flat mixed fields are shown to exhibit a  Heisenberg-type geometry, that is a
principal circle bundle over the torus $T^{2g}\times T^{2g}$ whose topology is governed by   
the standard symplectic form  multiplied by $\lambda$.  

We then turn to
T-duality, showing that it preserves the twisted Yang--Mills functional
and gives rise to a natural involution on a suitably defined
T-dualizable moduli stack of flat mixed fields. At $\lambda =1$, the stack collapses onto the moduli space 
and T-duality descends to an involutive contactomorphism of the moduli space itself.

Two natural directions remain: higher-rank torus fibrations, and the extension from $U(1)$ to $SU(N)$ twisted 
Yang--Mills theory, where spherical T-duality should provide the appropriate counterpart. Both will be pursued elsewhere.

\subsection{Main results} In the following, we describe the main results of this paper in more detail.

\subsubsection{Presymplectic reduction and the moduli space of flat mixed fields}

$\, $

A first basic result (Proposition~\ref{YMprop}) is 
\be \label{2YM}
\begin{split}
&\mathrm{YM}(A,B)\\
=&
\int_\Sigma |F_A|^2\,d\operatorname{vol}_h
+
\int_Z H\wedge *_{h_A}H \\
=&
\int_\Sigma
\left|
F_A-\frac{2\pi i\,\langle c_1(Z),[\Sigma]\rangle}
{\operatorname{Area}_h(\Sigma)}
\,d\operatorname{vol}_h
\right|^2
d\operatorname{vol}_h
+
\frac{4\pi^2}{\operatorname{Area}_h(\Sigma)}
\langle c_1(Z),[\Sigma]\rangle^2
+
\int_Z f^2\,d\operatorname{vol}_{h_A},
\end{split}
\ee
where
\[
H=f\,d\operatorname{vol}_{h_A}.
\]
The corresponding Euler--Lagrange equations are
\[
\left\{
\begin{aligned}
d_h^*F_A &=0,\\
d_{h_A}^*H &=0.
\end{aligned}
\right.
\]
It follows from \eqref{2YM} that when
\[
\langle c_1(Z),[\Sigma]\rangle=0,
\]
the absolute minima are precisely the flat mixed fields satisfying
\[
F_A=0,
\qquad
H=0.
\]
The
case of non-trivial degree,
\[
\langle c_1(Z),[\Sigma]\rangle\neq 0
\]
is also of interest, corresponding to {\em projectively flat mixed fields}, but we defer this to future study.

\medskip

The configuration space $\mathcal{C}$ in (\ref{conf}) carries, for each integer 
$\lambda\in\ZZ$, a natural {\em level-$\lambda$ action} of the group
\[
\GG_\lambda=\GG_1\ltimes_{\rho_\lambda}\GG_2 ,
\qquad
\GG_1=C^\infty(\Sigma,\mathbb T),\quad
\GG_2=\mathrm{Pic}^{\nabla}_{\mathbb T}(Z),
\]
where $C^\infty(\Sigma,\mathbb T)$ denotes the group of smooth circle-valued functions on $\Sigma$, 
while $\mathrm{Pic}^{\nabla}_{\mathbb T}(Z)$ is the Picard group of $\mathbb T$-equivariant line bundles 
over $Z$ equipped with $\mathbb T$-invariant connections. The  integrality of $\lambda$ is exactly what 
makes the semi-direct group law close (Lemma~\ref{lem:semidirect}), and we shall refer to $\GG$ as the {\em gauge group of mixed fields}.
The corresponding Lie algebras are 
\[
\mathfrak g_1:=\mathrm{Lie}(\GG_1)
=
C^\infty(\Sigma,\RR),
\]
where the infinitesimal action of
$f\in C^\infty(\Sigma,\RR)$
on $\Omega^1(\Sigma)$ is given by $df$, and by Section~5 of \cite{Kleiman},
\[
\mathfrak g_2:=\mathrm{Lie}(\GG_2)
=
\Omega^1(Z)^{\bbT},
\]
the space of $\bbT$-invariant $1$-forms on $Z$. For
$\omega\in\Omega^1(Z)^{\bbT}$,
the induced infinitesimal action on
$\Omega^2(Z)^{\bbT}$ is given by $d\omega$.

A first indication that $\GG_\lambda$ is the correct symmetry group is provided by Proposition~\ref{YMprop}, 
where we prove that the twisted Yang--Mills functional is invariant under its action.
The terminology ``gauge group'' is further justified in Section~\ref{gaugeloop}. There we show that the action of
$\mathrm{Pic}^{\nabla}_{\mathbb T}(Z)$
on $B$-fields admits a natural interpretation in loop space geometry. More precisely, it corresponds to 
gauge transformations of the connections on the holonomy line bundles over the free loop space $LZ$ 
associated to the $B$-fields on $Z$. 

\medskip

We next return to the geometry of the configuration space itself. In Section~\ref{geomstructure}, we 
construct a natural presymplectic form $\Omega$ on $\mathcal C$ which is degenerate,  but 
nevertheless retains enough structure to support a Hamiltonian description of the gauge action.
Indeed, the action of $\GG_\lambda$ is Hamiltonian, with moment map
\[
\mu(A,B)
=
\int_\Sigma F_A(\cdot)
+
\int_\Sigma \iota_v H\wedge(\iota_v\cdot)
\in
\mathfrak g_1^*\oplus\mathfrak g_2^*,
\]
as shown in Proposition~\ref{prop:moment}. The twisted Yang--Mills functional may be recovered directly 
from the moment map, namely with respect to the natural $L^2$-metric, Proposition~\ref{YMmoment} establishes the identity
\[
\mathrm{YM}(A,B)
=
|\mu(A,B)|^2.
\]

 The {\em moduli space of flat mixed fields} is defined by the presymplectic reduction
\[
\mu^{-1}(0)/\GG_\lambda,
\]
where $\mu^{-1}(0)$ consists precisely of those configurations satisfying
\[
F_A=0,
\qquad
H=0.
\] 

\medskip

\noindent\textbf{Theorem A} (Theorems~\ref{main1}, \ref{thm:cohomology} and~\ref{thm:cohomology-K}).  Let $\lambda\in\ZZ\setminus\{0\}$.
\begin{enumerate} 
\item[$(i)$] 
The moduli space  $\mu^{-1}(0)/\GG_\lambda$ is a  Heisenberg contact manifold of dimension $4g+1$, 
that is a principal circle bundle over the torus $T^{2g}\times T^{2g}$ with first Chern class
\[
c_1\bigl(\mu^{-1}(0)/\GG_\lambda\bigr)
=
-\lambda\sum_{j,k} (\Theta_g)_{jk}\,\alpha_j\cup\beta_k,
\]
where
\[
\Theta_g=
\begin{pmatrix}
0&I_g\\
-I_g&0
\end{pmatrix}
\]
is the standard symplectic matrix in $Sp(2g,\RR)$, $\alpha_j \in H^1(T^{2g}, \mathbb{Z})$ is the dual basis to 
the coordinate $x_j$ and $\beta_k \in H^1(T^{2g}, \mathbb{Z})$ is the dual basis to the coordinate $y_k$.

\item[$(ii)$] 
The fundamental group of $\mu^{-1}(0)/\GG_\lambda$ is a $\ZZ$ central extension of $\ZZ^{2g}\times \ZZ^{2g}$,
\[
\pi_1\bigl(\mu^{-1}(0)/\GG_\lambda\bigr)
=
(\ZZ^{2g}\times \ZZ^{2g})\times_{c_\lambda} \mathbb Z,
\]
with the group law 
\[
(\mathbf m,\mathbf n,\ell)(\mathbf m',\mathbf n',\ell')
=
\bigl(
\mathbf m+\mathbf m',
\mathbf n+\mathbf n',
\ell+\ell'+c_\lambda((\mathbf m,\mathbf n),(\mathbf m',\mathbf n'))
\bigr)
\]
and $2$-cocycle   
\[
c_\lambda((\mathbf m,\mathbf n),(\mathbf m',\mathbf n'))
=
\mathbf n' \Theta_g \mathbf m^T
+
(\lambda-1)\,\mathbf m' \Theta_g \mathbf n^T.
\]

\item[$(iii)$]  The integral cohomology and
$K$-theory groups are given by
\[
H^k\bigl(\mu^{-1}(0)/\GG_\lambda,\mathbb Z\bigr)
\cong
\mathbb Z^{\,b_k}\oplus\bigoplus_{j\ge1}\mathbb Z_{|\lambda|j}^{\,\varepsilon_{k-2}(j)}, 
\]
and
\[
K^{n}(\mu^{-1}(0)/\GG_\lambda)\cong \Z^{\,\binom{4g+1}{2g}}\oplus \bigoplus_{\substack{k\equiv_2 n\\ j\ge1}} \Z_{|\lambda|j}^{\,\varepsilon_{k-2}(j)},
\] 
with the Betti numbers  
\[b_k=
\begin{cases}
\binom{4g}{k}-\binom{4g}{k-2}, & 0\le k\le 2g,\\[2mm]
\binom{4g}{k-1}-\binom{4g}{k+1}, & 2g+1\le k\le 4g+1,
\end{cases}
\] and torsion multiplicities \[
\varepsilon_r(j)=\binom{4g}{r+2-2j}-\binom{4g}{r-2j}
\]
for $0\le r\le 2g-1$, and extended to $2g\le r\le 4g-2$ by $\varepsilon_r(j)=\varepsilon_{4g-2-r}(j)$. 
\end{enumerate}
\medskip

The   cohomology of Heisenberg manifolds was computed by Lee and Packer in \cite[Thm.~2.1]{LP} and the $K$-theory at $\lambda=1$ is
\cite[Thm.~3.9]{ALP}. What the present paper contributes is the $K$-theory for $|\lambda|\ge2$, which to our 
knowledge does not appear in the literature and which follows from the $\lambda$-uniform integral intertwiner of Lemma~\ref{lem:conj}. 

At $\lambda =0$, the gauge group becomes a product $\GG_1\times\GG_2$ and each factor acts independently 
on the mixed fields, so the moduli space degenerates to a torus $T^{4g+1}$.

\subsubsection{T-duality of the twisted Yang--Mills functional and the T-dualizable moduli stack}

$\, $
In T-duality, the radius of the circle fibre is transformed under the Buscher rules. Hence, we consider the family of
metrics  
\[
h_{A,R}
=
\pi^*h_X+R^2\,A\odot A,
\qquad
\widehat h_{\widehat A,1/R}
=
\widehat\pi^*h_X+\frac{1}{R^2}\,\widehat A\odot\widehat A,
\]
where $h_X$ is a Riemannian metric on $X$. Under T-duality, the fibre
radius is exchanged according to
\begin{center}
$R\Longleftrightarrow 1/R$.
\end{center}
It is therefore natural to ask whether the twisted Yang--Mills functional introduced above is compatible with this symmetry. 
To incorporate the radius parameter into the picture, we consider the rescaled functional
\[
\mathrm{YM}(A,B)_R
:=
\sqrt{R}\left(
\int_X |F_A|^2d\operatorname{vol}_h
+
\int_Z H\wedge *_{h_{A,R}}H
\right).
\]

The volume of the circle fibre scales linearly with $R$, while the Hodge star operator acting on forms with a vertical component 
introduces an additional $R$-dependence. The factor $\sqrt{R}$ precisely compensates for these scaling effects.
The resulting functional enjoys the expected transformation law under T-duality.

\medskip

\noindent\textbf{Theorem B} (Theorem~\ref{mainYM}). 
\[
\mathrm{YM}(A, B)_R
=
\mathrm{YM}(\widehat A, \widehat B)_{1/R}.
\]

\medskip
\medskip

From the discussion in Section~\ref{rev}, it follows that the T-duality transformation restricts to a self-map
\[
\mathcal T\colon \mu^{-1}(0)\longrightarrow\mu^{-1}(0).
\]
This suggests that one should look for  gauge transformations that are compatible with T-duality.
Accordingly, we make the following definition. A gauge transformation
\[
(g,(L,\nabla^{L}))\in\GG_\lambda
\]
is said to be {\em T-dualizable} at
\[
(A,B)\in\mu^{-1}(0)
\]
if there exists a dual pair
\[
(\widehat g,(\widehat L,\nabla^{\widehat L}))
\]
such that
\[
\mathcal T \left((g,(L,\nabla^{L}))\cdot(A,B)\right)
=
(\widehat g,(\widehat L,\nabla^{\widehat L}))
\cdot
\mathcal T(A,B).
\]
The collection of all T-dualizable gauge transformations at $(A,B)$ forms a  subgroup
\[
\GG^{T}_{(A,B)}
:=
\Bigl\{
(g,(L,\nabla^{L}))\in\GG_\lambda
\;\Big|\;
(g,(L,\nabla^{L}))
\text{ is T-dualizable at }(A,B)
\Bigr\}
<\GG_\lambda.
\]
Let $ h\in\GG^{T}_{(A,B)}$, then  $h\cdot(A,B)\in\mu^{-1}(0)$  and
$
\GG^{T}_{h\cdot(A,B)}
=
\GG^{T}_{(A,B)}.$
This observation allows us to introduce an equivalence relation on $\mu^{-1}(0)$ by declaring
\[
(A,B)\underset{T}{\sim}(A',B')
\quad\Longleftrightarrow\quad
(A',B')
=
h\cdot(A,B)
\text{ for some }
h\in\GG^{T}_{(A,B)}.
\]
The quotient
\[
\mu^{-1}(0)\big/\underset{T}{\sim}
\]
will be called the {\em T-dualizable moduli stack of flat mixed fields}.

\medskip
We can describe $\GG^{T}_{(A,B)}$ and the moduli stack more explicitly.  
Let $(\mathbf{x},\mathbf{y})
\in
\mathbb R^{2g}\times\mathbb R^{2g}
$ and define
\[
G_{(\mathbf{x},\mathbf{y})}
:=
\left\{
(\mathbf m,\mathbf n)
\in
\mathbb Z^{2g}\times\mathbb Z^{2g}
\;\Big|\;
(\lambda-1)\bigl(
\mathbf n\Theta_g\mathbf x^T
-
\mathbf m\Theta_g\mathbf y^T
\bigr)
\in\mathbb Z
\right\}.
\]
A direct verification shows that
\[
G_{(\mathbf{x},\mathbf{y})}
=
G_{(\mathbf{x}+\mathbf a,\mathbf y+\mathbf b)},
\qquad
(\mathbf a,\mathbf b)
\in
\mathbb Z^{2g}\times\mathbb Z^{2g}.
\]
Using $G_{(\mathbf x,\mathbf y)}$ we introduce a relation $\approx_T$ on
$\mathbb R^{2g}\times\mathbb R^{2g}\times\RR/\ZZ$ by declaring
\[
(\mathbf x,\mathbf y,[t])
\approx_T
(\mathbf x',\mathbf y',[t'])
\quad\Longleftrightarrow\quad
\mathbf x'=\mathbf x+\mathbf m,\ \
\mathbf y'=\mathbf y+\mathbf n,\ \
[t']=\bigl[t+\phi^{\lambda}_{(\mathbf m,\mathbf n)}(\mathbf x,\mathbf y)\bigr]
\]
for some $(\mathbf m,\mathbf n)\in G_{(\mathbf x,\mathbf y)}$, where
$\phi^{\lambda}_{(\mathbf m,\mathbf n)}(\mathbf x,\mathbf y)
=\mathbf n\Theta_g\mathbf x^T+(\lambda-1)\,\mathbf m\Theta_g\mathbf y^T$
is the level-$\lambda$ fibre shift.  
\newpage

\noindent\textbf{Theorem C} (Theorem~\ref{main2}).  
\begin{enumerate}
\item[$(i)$] 
For every \( (A,B)\in\mu^{-1}(0) \),
\[
\GG^{T}_{(A,B)}
\;\cong\;
\GG_{0}\oplus \ker l_{(A,B)},
\]
where
\[
\GG_{0}
:=
\left\{
(g,(L,\nabla^{L}))\in\GG
\;\Big|\;
g=e^{2\pi i f_{1}},\;
\mathrm{hol}^{v}(\nabla^{L})=e^{2\pi i f_{2}},\;
f_{1},f_{2}\in C^\infty(\Sigma,\RR)
\right\}
<\GG_\lambda,
\]
and
\[
l_{(A,B)} :
\mathbb{Z}^{2g}\times\mathbb{Z}^{2g}
\longrightarrow
\RR/\ZZ,
\qquad
(\mathbf m,\mathbf n)
\longmapsto
\bigl\{
(\lambda-1)\bigl(
\mathbf n \Theta_g [a_{0}]^{T}
-
\mathbf m \Theta_g [b_{0}]^{T}
\bigr)
\bigr\},
\]
is a homomorphism. Here
\(a_0\in \Omega^{1}_{\mathrm{cl}}(\Sigma)\) and
\(b_0\in \Omega^{1}_{\mathrm{cl}}(\Sigma)\)
are differential forms canonically determined by \(A\) and \(B\)
(see~\eqref{expression}), while
\([a_0]\) and \([b_0]\) denote the corresponding vectors in
\(H^1(\Sigma,\RR)\cong\RR^{2g}\) after fixing a symplectic basis.

\item[$(ii)$] 
The relation $\approx_T$ is an equivalence relation, and there is a natural bijection
\[
\mu^{-1}(0)\big/\underset{T}{\sim}
\;\cong\;
\bigl(\RR^{2g}\times\RR^{2g}\times \RR/\ZZ\bigr)\big/\approx_T .
\]

\item[$(iii)$] 
The T-duality map $\mathcal T:\mu^{-1}(0)\longrightarrow\mu^{-1}(0)$
induces a well-defined involutive transformation
\[
\mathcal T_{*}:
\mu^{-1}(0)\big/\underset{T}{\sim}
\longrightarrow
\mu^{-1}(0)\big/\underset{T}{\sim}, 
\]
which, under the identification in $(ii)$, reads
\[
\mathcal T_{*}\bigl[(\mathbf x,\mathbf y),[t]\bigr]
=
\bigl[
(-\mathbf y,-\mathbf x),
\,[t+\mathbf x\Theta_g\mathbf y^{T}]
\bigr].
\]
\end{enumerate}

\medskip

There is a canonical   surjection from the T-dualizable moduli stack to the   moduli space of flat mixed fields, 
\[p\colon \mu^{-1}(0)\big/\underset{T}{\sim} \ \to \ \mu^{-1}(0)\big/\GG_\lambda \]
where the pre-image 
\[ p^{-1}(A,B) \cong \mathbb{G}_\lambda / \mathbb{G}_{(A,B)}^T \cong \operatorname{im}(l_{(A,B)}) \]
 is determined by the action of the full gauge group $\mathbb{G}_\lambda$ relative to the T-dualizable subgroup $\mathbb{G}_{(A,B)}^T$.
Hence, along $p$ the $\RR/\ZZ$ fibers of the two quotients are related by a fiber-wise map whose kernel is precisely the 
finitely generated subgroup $\operatorname{im}(l_{(A,B)}) \subset \RR/\ZZ$. Over generic points in the base 
$T^{2g}\times T^{2g}$, $\operatorname{im}(l_{(A,B)})$ is a dense subgroup of rank $4g$, making $p^{-1}(A,B)$ 
countably infinite, while over rational points in $T^{2g}\times T^{2g}$, $\operatorname{im}(l_{(A,B)})$ is a finite cyclic group, and the fiber collapses.

We conclude the paper by presenting the moduli space as an action Lie groupoid $\mathcal H$ carrying a canonical
$\mathbb{R}/\mathbb{Z}$-valued cocycle $\Phi$.  The kernel of $\Phi$ is a Lie groupoid $\mathcal K$ that presents the moduli stack.  
T-duality acts on $\mathcal H$, but descends to a genuine groupoid automorphism only on $\mathcal K$.

\subsection*{Organization of the paper}
In Section~\ref{geomstructure}, we introduce the configuration space of mixed fields and its gauge symmetry, define the 
twisted Yang--Mills functional, and establish its basic geometric properties. In particular, we construct a natural 
presymplectic structure on the configuration space and identify the corresponding moment map.

Section~\ref{gaugeloop} provides an interpretation of the $B$-field gauge symmetry in terms of holonomy line bundles over the free loop space.

In Section~\ref{moduli}, we perform the presymplectic reduction of the configuration space and determine the topology and 
geometry of the moduli space of flat mixed fields, proving Theorem~A.

The second part of the paper is devoted to the interaction of these constructions with T-duality. Section~\ref{rev} reviews the 
geometric framework of T-duality for principal circle bundles with flux. In Section~\ref{TYM}, we establish the T-duality 
invariance of the suitably rescaled twisted Yang--Mills functional, proving Theorem~B.  Section~\ref{Tmoduli} introduces the 
T-dualizable moduli stack and studies the induced action of T-duality on this quotient, leading to Theorem~C. 
Finally, Section~\ref{Groupoid} gives  a groupoid presentation of the moduli stack and describes the action of T-duality.

\subsection*{Acknowledgement}  Fei Han was partially supported by the AcRF grant from the National University of Singapore. 
Tsuyoshi Kato is supported by JSPS KAKEHI Grant number 23K22394. Varghese Mathai thanks the Australian Research Council  
for support via the Australian Laureate Fellowship FL170100020. Pedram Hekmati is grateful to Max Planck Institute for Mathematics in
Bonn for its hospitality and financial support.


\section{The configuration space of mixed fields} \label{geomstructure}

\subsection{The configuration space of mixed fields and gauge transformations} 
$\, $

Let $\Sigma$ be a closed oriented $2$-manifold of genus $g$. Fix a Riemannian metric $h$ on $\Sigma$ together with a good open cover $\mathcal U=\{U_\alpha\}$.
Let $Z$ be a principal $\bbT$-bundle over $\Sigma$,
\[
\begin{CD}
\bbT @>>> Z \\
@.   @V\pi VV \\
@.   \Sigma
\end{CD}
\]
The metric $h$ on the base, together with the standard metric on the circle, induces a natural metric on the total space $Z$.
Let $A$ be a connection on $Z$ and let $B$
be a $\TT$-invariant $B$-field on $Z$. Motivated by  topological T-duality reviewed in Section~\ref{rev}, we consider the \emph{configuration space}
\[
\Ca=\set{(A,B)}.
\]

Let us first  clarify what is meant by $B$, since the notation $(A,B)$ suppresses some data.
A $\bbT$-invariant $B$-field is a {\em curving} for a $\bbT$-invariant gerbe on $Z$. Relative to the cover
$\{U_\alpha\times\bbT\}$ it consists of local $2$-forms $B_\alpha$ subject to
\be\label{eq:curving}
B_\beta-B_\alpha=dA_{\alpha\beta},
\ee
\be\label{eq:triple}
A_{\alpha\beta}+A_{\beta\gamma}+A_{\gamma\alpha}=g_{\alpha\beta\gamma}^{-1}dg_{\alpha\beta\gamma},
\ee
where $(A_{\alpha\beta},g_{\alpha\beta\gamma})$  is part of the connective structure on the gerbe,  see Section~\ref{gaugeloop}. 
The class of the $\bbT$-invariant \v Cech $2$-cocycle $g_{\alpha\beta\gamma}$  in
$H^2(Z;\underline{U(1)})\cong H^3(Z,\ZZ)$ is the Dixmier--Douady class of the gerbe.
Throughout we write simply $B$ for the curving and suppress the cochain
$(A_{\alpha\beta},g_{\alpha\beta\gamma})$ from the notation, except at places  where it matters. This is   harmless 
since the flux $H=dB$ is a global $3$-form 
and both the presymplectic form \eqref{eq:Omega-def} and the moment map \eqref{moment} are built
from $H$ and from deformations $\delta B$, which by the discussion below are global. 

With this understood, we fix the cohomology class of the cocycle $g_{\alpha\beta\gamma}$ (i.e. the $H$-flux), and the 
configuration space is  the affine space 
\be
\Ca=\bigl\{\,(A,B)\ :\ B=\{B_\alpha, A_{\alpha\beta}\}\ \text{a $\bbT$-invariant cochain satisfying
\eqref{eq:curving} and \eqref{eq:triple}}\,\bigr\}.
\ee
Consider the vector space of $\bbT$-invariant cochains,
\[V:=\bigl\{\,(b_\alpha,a_{\alpha\beta})\ :\ b_\beta-b_\alpha=da_{\alpha\beta},\ \
a_{\alpha\beta}+a_{\beta\gamma}+a_{\gamma\alpha}=0\,\bigr\}.
\]
Then $\Ca$  is a torsor over  $\Omega^1(\Sigma)\times V$,
the first factor acting on $A$ and the second on $(B_\alpha,A_{\alpha\beta})$. 
Indeed, if
$(B_\alpha,A_{\alpha\beta})$ and $(B'_\alpha,A'_{\alpha\beta})$ are two configurations with the
same $A$, their difference $(b_\alpha,a_{\alpha\beta})=(B'_\alpha-B_\alpha,
A'_{\alpha\beta}-A_{\alpha\beta})$ satisfies
$b_\beta-b_\alpha=dA'_{\alpha\beta}-dA_{\alpha\beta}=da_{\alpha\beta}$ by
\eqref{eq:curving}, and $\sum_{\rm cyc}a_{\alpha\beta}=0$ by \eqref{eq:triple} with the term
$g_{\alpha\beta\gamma}^{-1}dg_{\alpha\beta\gamma}$ cancelling because
$g_{\alpha\beta\gamma}$ is the same for both, so the difference lies in $V$. Conversely, adding
$(b_\alpha,a_{\alpha\beta})\in V$ to a configuration preserves \eqref{eq:curving} and
\eqref{eq:triple}, so the action  is free and transitive.  We also note that the sub-torsor at fixed gluing data, on which
$a_{\alpha\beta}=0$ and $b_\alpha$ is a global form, is a torsor over the space of \(\bbT\)-invariant \(2\)-forms
$\Omega^2(Z)^\bbT$.

It will be convenient to record the local decomposition
\be\label{eq:local-decomp}
B_\alpha=b_{1,\alpha}+A\wedge b_{0,\alpha},
\qquad
b_{1,\alpha}\in\Omega^{2}(U_\alpha),
\quad
b_{0,\alpha}\in\Omega^{1}(U_\alpha),
\ee
which separates the part of $B_\alpha$ with a vertical leg from the part without. On a single
chart the decomposition is a canonical isomorphism
$\Omega^2(U_\alpha\times\bbT)^\bbT\cong\Omega^2(U_\alpha)\oplus\Omega^1(U_\alpha)$, where
$b_{0,\alpha}=\iota_vB_\alpha$ and $b_{1,\alpha}=B_\alpha-A\wedge b_{0,\alpha}$, the latter being
basic since $\iota_vb_{1,\alpha}=\iota_vB_\alpha-b_{0,\alpha}=0$,  and $\iota_v$ denoting the contraction with the 
fibre-generating vector field. The decomposition is however not a product
$\prod_\alpha\Omega^2(U_\alpha)\times\prod_\alpha\Omega^1(U_\alpha)$ since  $\{B_\alpha\}$ is subject to the gluing
\eqref{eq:curving}, which by \eqref{eq:local-decomp} reads
\[
b_{0,\beta}-b_{0,\alpha}=-d(\iota_vA_{\alpha\beta}),
\qquad
b_{1,\beta}-b_{1,\alpha}=dA_{\alpha\beta}+A\wedge d(\iota_vA_{\alpha\beta}),
\]
so the local pieces are tied together across overlaps. On the flat locus of
Section~\ref{moduli}, where $Z$ is trivialized and the gerbe data may be taken to vanish, the
decomposition becomes global and will be used in that form in Lemma~\ref{lem:finite-model}.

The deformation of a connection $A$ does not change the cohomology class represented by its curvature.
{\em For the $B$-field, we restrict to deformations that preserve the
gluing data $A_{\alpha\beta}$, or equivalently, we deform the curving only}. Such a deformation
satisfies $\delta B_\beta-\delta B_\alpha=0$ by \eqref{eq:curving}, and is therefore a global
$\bbT$-invariant $2$-form on $Z$. The transgression of
such a $2$-form gives the corresponding deformation of the connection on the holonomy line bundle
over the loop space associated to the $B$-field; see Remark~\ref{add2form}. Therefore
the tangent space at a point \((A,B)\) is
\[
T_{(A,B)}\Ca
\cong
\Omega^1(\Sigma)\times \Omega^2(Z)^{\TT}.
\]
 
\medskip

Let $\GG_1:=C^\infty(\Sigma,\bbT)$ be the gauge group for connections. For $g\in \GG_1$,
\[
g\cdot A \;=\; A + d\ln g .
\]
Let $\GG_2:=\pic_{\bbT}(Z)$ be the Picard group of $\bbT$-equivariant line bundles with $\bbT$-invariant connections over $Z$.
Its action on $B$-fields is
\[
(L,\nabla^L)\cdot B \;=\; \{\,B_\alpha + (\nabla^L)^2\,\}.
\]
We call $\GG_2$ the \emph{gauge  group of the $B$-fields}.
In Section~\ref{gaugeloop} we  justify this terminology by relating the Picard group action on $B$-fields
to gauge transformations of line bundle connections over free loop spaces (see Proposition~\ref{loop gauge}).
The relation between gerbes, holonomy and transgression can be found
in~\cite{Bry}.

Throughout, we use the notation
\[
u:=g^{-1}dg,
\qquad
\eta:=F_{\nabla^L}=(\nabla^L)^2,
\]
for \(g\in\GG_1\) and
\(\mathcal L=(L,\nabla^L)\in\GG_2\). We shall make no distinction
between a form on \(\Sigma\) and its pullback to \(Z\). Thus \(u\) is
a closed \(1\)-form on \(\Sigma\) with integral periods, while \(\eta\)
is a closed \(\bbT\)-invariant \(2\)-form on \(Z\), again with integral
periods.

The vertical holonomy of \(\nabla^L\) is a \(\bbT\)-valued function
\[
\hol^v(\nabla^L)\colon \Sigma\longrightarrow \bbT,
\]
and one has
\be\label{eq:iota-eta}
\iota_v\eta
=
-\,d\log\hol^v(\nabla^L).
\ee
In particular, \(\iota_v\eta\) is basic and descends to a closed
\(1\)-form on \(\Sigma\) with integral periods.

Interestingly, a $\TT$-invariant gerbe over a circle bundle over a Riemann surface admits a natural action of the 
gauge group $\GG_1$ by gerbe automorphisms, indexed by an integer $\lambda \in \Z $. These combine into 
an action by a semi-direct product $\GG_1 \ltimes_{\rho_\lambda} \GG_2$ on the configuration space as explained below.

\begin{definition}\label{def:twisted-action}
Let \(\lambda\in\ZZ\). For $(g, \mathcal L=(L,\nabla^L))\in \GG_1\times \GG_2$,
the associated \emph{level-\(\lambda\) transformation} of
\((A,B)\in\Ca\) is defined by
\be\label{gautran}
\begin{aligned}
\widetilde A
&=A+u,\\
\widetilde B_\alpha
&=B_\alpha+\eta
  +\lambda\,u\wedge\iota_vB_\alpha,\\
\widetilde A_{\alpha\beta}
&=A_{\alpha\beta}
  +\lambda\,(\iota_vA_{\alpha\beta})\,u.
\end{aligned}
\ee
We write this transformation as
\[
(g,\mathcal L)\cdot(A,B)
=
(\widetilde A,\widetilde B).
\]
\end{definition}
There are two elementary identities behind these formulas. Since \(u\)
is basic and \(\iota_v^2=0\),
\be\label{eq:iota-twist}
\iota_v\bigl(u\wedge\iota_vB_\alpha\bigr)
=
(\iota_vu)\,\iota_vB_\alpha
-u\wedge\iota_v^2B_\alpha
=
0,
\ee
and
\be\label{eq:iota-twist-A}
\iota_v\widetilde A_{\alpha\beta}
=
\iota_v\bigl(
A_{\alpha\beta}
+\lambda(\iota_vA_{\alpha\beta})u
\bigr)
=
\iota_vA_{\alpha\beta}.
\ee
Thus the $\lambda$-dependent terms are nilpotent in the simplest possible
sense, namely once inserted, they do not produce a second
correction of the same kind.

There is also a dimensional fact which will be used repeatedly.
Since the gerbe data is \(\bbT\)-invariant and 
\(\mathcal L_v=d\iota_v+\iota_vd=0\), we have
\[
d\iota_vB_\alpha=-\iota_vH.
\]
It follows that
\be\label{eq:d-twist}
d\bigl(u\wedge\iota_vB_\alpha\bigr)
=-u\wedge d(\iota_vB_\alpha)
=u\wedge\iota_vH =0.
\ee
Indeed, both \(u\) and \(\iota_vH\) are basic, so their wedge product
is the pullback of a \(3\)-form on the surface \(\Sigma\). This is the
only point at which the assumption \(\dim\Sigma=2\) enters, but in an essential way.

\begin{lemma}\label{lem:cocycle-repair}
Let \((A,B)\in\Ca\), with gerbe data
\[
(B_\alpha,A_{\alpha\beta},g_{\alpha\beta\gamma}),
\]
and let \((\widetilde A,\widetilde B)\) be its transform under
\eqref{gautran}. Then
\[
\widetilde B_\beta-\widetilde B_\alpha
=
d\widetilde A_{\alpha\beta},
\]
and
\[
\widetilde A_{\alpha\beta}
+\widetilde A_{\beta\gamma}
+\widetilde A_{\gamma\alpha}
=
g_{\alpha\beta\gamma}^{-1}dg_{\alpha\beta\gamma}.
\]
Thus the transformed local data again defines a \(B\)-field on \(Z\).
Moreover,
\[
\widetilde H=H,
\qquad
F_{\widetilde A}=F_A.
\]
In particular, the Dixmier--Douady class of the gerbe is unchanged.
\end{lemma}

\begin{proof}
Since all the local data is \(\bbT\)-invariant, applying \(\iota_v\) to the descent equation
\[
B_\beta-B_\alpha=dA_{\alpha\beta}
\]
yields
\[
\iota_vB_\beta-\iota_vB_\alpha
=
-d(\iota_vA_{\alpha\beta}).
\]
Put $h_{\alpha\beta}:=\iota_vA_{\alpha\beta}$.  Since \(du=0\), we have
\[
d(h_{\alpha\beta}u)
=
dh_{\alpha\beta}\wedge u
=
-u\wedge dh_{\alpha\beta}, 
\]
and it follows that
\[
\lambda u\wedge\iota_vB_\beta
-\lambda u\wedge\iota_vB_\alpha
=
-\lambda u\wedge dh_{\alpha\beta}\\
=
d(\lambda h_{\alpha\beta}u).
\]
The form \(\eta\) is global and disappears on taking
differences, so
\[
\begin{aligned}
\widetilde B_\beta-\widetilde B_\alpha
&=
dA_{\alpha\beta}
+d(\lambda h_{\alpha\beta}u)\\
&=
d\bigl(
A_{\alpha\beta}
+\lambda(\iota_vA_{\alpha\beta})u
\bigr)\\
&=
d\widetilde A_{\alpha\beta}.
\end{aligned}
\]
The correction in $\widetilde A_{\alpha\beta}$  is thus precisely the term required to
repair the descent equation.

On a triple overlap, we obtain from \eqref{eq:triple}
\[
\widetilde A_{\alpha\beta}
+\widetilde A_{\beta\gamma}
+\widetilde A_{\gamma\alpha}
=
g_{\alpha\beta\gamma}^{-1}dg_{\alpha\beta\gamma}  +\lambda\,
\iota_v\bigl(
g_{\alpha\beta\gamma}^{-1}
dg_{\alpha\beta\gamma}
\bigr)u.
\]
The transition function \(g_{\alpha\beta\gamma}\) is
\(\bbT\)-invariant. Hence
\[
\iota_v\bigl(
g_{\alpha\beta\gamma}^{-1}
dg_{\alpha\beta\gamma}
\bigr)
=
g_{\alpha\beta\gamma}^{-1}
v(g_{\alpha\beta\gamma})
=
0,
\]
and the triple-overlap condition is unchanged. Finally, using \eqref{eq:d-twist},
\[
\begin{aligned}
\widetilde H|_{U_\alpha}
&=
d\widetilde B_\alpha\\
&=
dB_\alpha+d\eta
+\lambda\,d(u\wedge\iota_vB_\alpha)\\
&=
H|_{U_\alpha}.
\end{aligned}
\]
Likewise, $F_{\widetilde A}=F_{A+u}=F_A,$ as \(du=0\).
\end{proof}

We now give a geometric interpretation of the group law encoded in
\eqref{gautran}. The mixed term which appears when two
transformations are composed is the curvature
of a natural line bundle on the base. In fact, let
\[
p_i\colon\bbT\times\bbT\longrightarrow\bbT,
\qquad i=1,2,
\]
denote the two projections. In our normalization, the curvature of the Poincaré line bundle $\mathcal{P}$ with standard connection over $\bbT\times\bbT$
  is
\[
F_{\nabla^{\mathcal P}}
=
p_1^*\vartheta\wedge p_2^*\vartheta.
\]
For \(\mathcal L=(L,\nabla^L)\), set
\[
h_{\mathcal L}
:=
\bigl(\hol^v(\nabla^L)\bigr)^{-1}.
\]
By \eqref{eq:iota-eta}, we have $h_{\mathcal L}^{-1}dh_{\mathcal L} = \iota_v\eta$.  Define the line bundle with connection
\[
(\xi_{g,\mathcal L},\nabla^{\xi_{g,\mathcal L}})
:=
(g,h_{\mathcal L})^*
(\mathcal P,\nabla^{\mathcal P})
\longrightarrow\Sigma.
\]
Its curvature is
\be\label{eq:xi-curvature}
F_{\nabla^{\xi_{g,\mathcal L}}}
=
u\wedge\iota_v\eta.
\ee
For \(\lambda\in\ZZ\), define
\be\label{action}
\rho_\lambda(g)(\mathcal L)
:=
\mathcal L\otimes
\pi^*
(\xi_{g,\mathcal L},
 \nabla^{\xi_{g,\mathcal L}})^{\otimes\lambda},
\ee
where negative tensor powers are understood as powers of the dual line
bundle. At the level of curvature forms,
\be\label{eq:rho-curvature}
F_{\rho_\lambda(g)(\mathcal L)}
=
\eta+\lambda\,u\wedge\iota_v\eta.
\ee

The Poincaré line bundle is multiplicative in each variable and gives canonical connection-preserving
isomorphisms
\[
\xi_{g,\mathcal L_1\otimes\mathcal L_2}
\cong
\xi_{g,\mathcal L_1}
\otimes
\xi_{g,\mathcal L_2}
\]
and
\[
\xi_{g_2g_1,\mathcal L}
\cong
\xi_{g_2,\mathcal L}
\otimes
\xi_{g_1,\mathcal L}.
\]
Moreover, tensoring \(\mathcal L\) with the pullback of a line bundle
from \(\Sigma\) does not alter its vertical holonomy. Hence
\[
\xi_{g_2,\rho_\lambda(g_1)(\mathcal L)}
\cong
\xi_{g_2,\mathcal L}.
\]
These identities show that
\[
\rho_\lambda(g_2g_1)
=
\rho_\lambda(g_2)\circ\rho_\lambda(g_1),
\]
and therefore define a homomorphism
\[
\rho_\lambda\colon
\GG_1\longrightarrow\operatorname{Aut}(\GG_2).
\]

At the level of curvature forms, the same fact is reflected in the
elementary identity
\be\label{eq:rho-composition}
\begin{aligned}
&\eta+\lambda u_1\wedge\iota_v\eta
+\lambda u_2\wedge
 \iota_v\bigl(
   \eta+\lambda u_1\wedge\iota_v\eta
 \bigr)\\
&\hspace{4em}
=
\eta+\lambda(u_1+u_2)\wedge\iota_v\eta,
\end{aligned}
\ee
where we have used $\iota_v(u_1\wedge\iota_v\eta)=0.$

We may therefore form the semi-direct product
\[
\GG_\lambda
:=
\GG_1\ltimes_{\rho_\lambda}\GG_2.
\]
Writing \(\mathcal L_i=(L_i,\nabla^{L_i})\), its multiplication is
\be\label{eq:semidirect-law}
(g_2,\mathcal L_2)(g_1,\mathcal L_1)
=
\bigl(
g_2g_1,\,
\mathcal L_2\otimes\rho_\lambda(g_2)(\mathcal L_1)
\bigr).
\ee
Equivalently,
\[
\begin{aligned}
&(g_2,(L_2,\nabla^{L_2}))
(g_1,(L_1,\nabla^{L_1}))\\
&\qquad
=
\Bigl(
g_2g_1,\,
(L_2,\nabla^{L_2})
\otimes
(L_1,\nabla^{L_1})
\otimes
\pi^*
(\xi_{g_2,\mathcal L_1},
 \nabla^{\xi_{g_2,\mathcal L_1}})^{\otimes\lambda}
\Bigr).
\end{aligned}
\]

\begin{lemma}\label{lem:semidirect}
For every \(\lambda\in\ZZ\), the transformations
\eqref{gautran} define an action of
\[
\GG_\lambda
=
\GG_1\ltimes_{\rho_\lambda}\GG_2
\]
on \(\Ca\).
\end{lemma}

\begin{proof}
Let
\[
(g_i,\mathcal L_i)\in\GG_\lambda,
\qquad
u_i=g_i^{-1}dg_i,
\qquad
\eta_i=F_{\nabla^{L_i}},
\quad i=1,2.
\]
The line-bundle component of the product
\((g_2,\mathcal L_2)(g_1,\mathcal L_1)\) has curvature
\be\label{eq:composite-eta}
\eta_{21}
=
\eta_2+\eta_1
+\lambda\,u_2\wedge\iota_v\eta_1.
\ee
The \(\GG_1\)-component has Maurer--Cartan form
\[
(g_2g_1)^{-1}d(g_2g_1)=u_1+u_2.
\]
Apply first \((g_1,\mathcal L_1)\), and then
\((g_2,\mathcal L_2)\). The ordinary connection transforms as
\[
A\longmapsto A+u_1+u_2.
\]
For the local \(2\)-forms, the first transformation gives
\[
B_\alpha'
=
B_\alpha+\eta_1
+\lambda u_1\wedge\iota_vB_\alpha.
\]
By \eqref{eq:iota-twist},
\[
\iota_vB_\alpha'
=
\iota_vB_\alpha+\iota_v\eta_1.
\]
The second transformation therefore gives
\[
\begin{aligned}
B_\alpha''
&=
B_\alpha'
+\eta_2
+\lambda u_2\wedge\iota_vB_\alpha'\\
&=
B_\alpha+\eta_1+\eta_2
+\lambda(u_1+u_2)\wedge\iota_vB_\alpha
+\lambda u_2\wedge\iota_v\eta_1\\
&=
B_\alpha+\eta_{21}
+\lambda(u_1+u_2)\wedge\iota_vB_\alpha.
\end{aligned}
\]
This is exactly the transformation associated with the product
\((g_2,\mathcal L_2)(g_1,\mathcal L_1)\).

For the overlap \(1\)-forms, \eqref{eq:iota-twist-A} gives
\[
\iota_vA_{\alpha\beta}'
=
\iota_vA_{\alpha\beta}.
\]
Hence
\[
\begin{aligned}
A_{\alpha\beta}''
&=
A_{\alpha\beta}
+\lambda(\iota_vA_{\alpha\beta})u_1
+\lambda(\iota_vA_{\alpha\beta})u_2\\
&=
A_{\alpha\beta}
+\lambda(\iota_vA_{\alpha\beta})(u_1+u_2),
\end{aligned}
\]
which again agrees with the transformation defined by the product
element. Thus \eqref{gautran} respects the multiplication
\eqref{eq:semidirect-law}, and hence defines a group action.

For completeness, the inverse is
\[
(g,\mathcal L)^{-1}
=
\bigl(
g^{-1},
\rho_\lambda(g^{-1})(\mathcal L^\vee)
\bigr).
\]
If \(u=g^{-1}dg\) and \(F_{\nabla^L}=\eta\), then the curvature of its
line-bundle component is
\[
-\eta+\lambda\,u\wedge\iota_v\eta.
\]
\end{proof}

We call $\GG_\lambda$ {\em the gauge group of the mixed fields}. The integrality of $\lambda$ is exactly what is
needed for the group to close and is the source of the quantisation of the level. The action $\rho_\lambda$ of $\GG_1$ on $\GG_2$ is 
infinite-dimensional origin of the Heisenberg extension found in Section~\ref{moduli}.

\medskip

The following proposition shows that, once one asks for a semi-direct product  action  built naturally from the fields, \eqref{gautran} is  the only
possibility.

\begin{proposition}\label{prop:c-uniqueness}
Let $c : \GG_1 \times \Ca \to \GG_2$ be a map and suppose $c(g,A, B)$ depends
on $g \in \GG_1$   $\RR$-linearly through $u = g^{-1}dg$, and is
constructed from $u$, $A$ and $B$ using only the operations $\wedge$, $d$ and
$\iota_v$. Suppose further that
\[
  (A,B) \;\longmapsto\; \bigl(A+u,\; B+\eta+c(g,A,B)\bigr)
\]
defines an action of a semi-direct product $\GG_1 \ltimes \GG_2$ on $\Ca$. Then
\[
  c(g,A,B) = \lambda\, u \wedge \iota_v B,\qquad \lambda\in\ZZ .
\]
\end{proposition}

\begin{proof}
Linearity in $u$ and the fact that $c$ has degree $2$ force $c(g,A,B)=u\wedge X$,
with $X$ a $1$-form built from $u$, $A$ and $B$ using $\wedge$, $d$ and $\iota_v$.
Since $u$ is
basic, $\iota_v u=0$ and $u\wedge u=0$, so any $u$-dependent term of $X$
does not contribute to $u\wedge X$, and we may take $X$ built from $A$ and $B$
alone. Applying the operations to $A$ and $B$ produces $A$, $B$, $H=dB$,
$F_A=dA$, $\iota_vB$, $\iota_vH$, $d\iota_vB$, together with $\iota_vA=1$ and the
vanishing combinations $\iota_v\iota_vB=0$ and
$\iota_vF_A=\iota_v dA=-\,d(\iota_vA)=0$. Contractions of mixed wedge products
reduce to these, for instance $\iota_v(A\wedge\iota_vB)=(\iota_vA)\,\iota_vB
-A\wedge\iota_v\iota_vB=\iota_vB$. Of the resulting nonzero forms, only
$\iota_vB$ and $A$ have degree $1$. Hence $X=\lambda\,\iota_vB+\kappa\,A$ for
scalars $\lambda,\kappa$, that is
\[
  c(g,A,B)=\lambda\,u\wedge\iota_vB+\kappa\,u\wedge A .
\]

Composing two transformations, $(g_1,\eta_1)$ and then
$(g_2,\eta_2)$, the second uses the already-transformed connection $A+u_1$, so
the $B$-component acquires
\[
  \bigl(\eta_1+\eta_2\bigr)
  +\lambda\bigl(u_1+u_2\bigr)\wedge\iota_vB
  +\lambda\,u_2\wedge\iota_v\eta_1
  +\kappa\bigl(u_1+u_2\bigr)\wedge A
  +\kappa\,u_2\wedge u_1 ,
\]
where we used $\iota_v(u_1\wedge\iota_vB)=0$ from \eqref{eq:iota-twist}. The
terms proportional to $(u_1+u_2)$ are precisely the correction for the composite
parameter, and $\lambda\,u_2\wedge\iota_v\eta_1$ is the twist of
Lemma~\ref{lem:semidirect}. The remaining term $\kappa\,u_2\wedge u_1$ is a
closed $2$-form on $\Sigma$ that pairs two elements of $\GG_1$, so it is the
$2$-cocycle of a central extension of $\GG_1$ by $\GG_2$, rather than a term in
an action of $\GG_1$ on $\GG_2$. To obtain a semi-direct product
$\GG_1\ltimes\GG_2$, we must therefore have $\kappa=0$.

The  term
$\lambda\,u_2\wedge\iota_v\eta_1$ is a closed basic $2$-form whose period
$\lambda\int_\Sigma u_2\wedge\iota_v\eta_1$ must be integral for the composite to
lie in $\GG_2$. As $u_2$ and $\iota_v\eta_1$ have integral periods, this forces
$\lambda\in\ZZ$.
\end{proof}

Clearly,
\[
\mathfrak g_1 := \mathrm{Lie}(\GG_1) = C^\infty(\Sigma,\RR),
\]
and for $f\in C^\infty(\Sigma,\RR)$ the associated Killing vector field on $\Ca$ is
\be\label{eq:killing-1}
\rho(f)=\bigl(df,\ \lambda\,df\wedge\iota_vB\bigr)\in \Omega^1(\Sigma)\times\Omega^2(Z)^\bbT ,
\ee
together with the deformation $\lambda\,(\iota_vA_{\alpha\beta})\,df$ of the gluing data. Moreover (cf.~Sec.~5 in \cite{Kleiman}),
\[
\mathfrak g_2 := \mathrm{Lie}(\GG_2) = \Omega^1(Z)^{\bbT},
\]
the space of $\bbT$-invariant $1$-forms on $Z$. For $\omega\in  \Omega^1(Z)^{\bbT}$, the induced Killing vector field on
$\Omega^2(Z)^{\bbT}$ is $d\omega$.

 \subsection{Twisted Yang-Mills functional}

Motivated by Segal's  paper \cite{Segal}, we introduce the {\em twisted Yang-Mills functional}
\[ \mathrm{YM}(A, B):=\int_\Sigma |F_A|^2\,d\operatorname{vol}_h+ \int_Z H \wedge *_{h_A} H,\]
where the twisted term $\int_Z H \wedge *_{h_A} H$ is added into the usual Yang-Mills functional and measures 
the $L^2$-energy of the $H$-flux with respect to the metric $h_A$. Thus the functional couples the gauge field $A$
and the $B$-field through the geometry of the total space $Z$.

\begin{proposition}\label{YMprop}
\begin{enumerate}[(i)]
\item The twisted Yang-Mills functional $\mathrm{YM}(A,B)$ is gauge invariant.

\item One has the identity
\begin{equation}\label{22YM}
\begin{aligned}
&\mathrm{YM}(A,B)\\
=&
\int_\Sigma
\left|
F_A-\frac{2\pi i\,\langle c_1(Z),[\Sigma]\rangle}{\operatorname{Area}_h(\Sigma)}
\,d\operatorname{vol}_h
\right|^2
\,d\operatorname{vol}_h \\
\quad&
+
\frac{4\pi^2}{\operatorname{Area}_h(\Sigma)}
\langle c_1(Z),[\Sigma]\rangle^2
+
\int_Z f^2\,d\operatorname{vol}_{h_A},
\end{aligned}
\end{equation}
where $H=f\,d\operatorname{vol}_{h_A}$.

\item The Euler--Lagrange equations of $\mathrm{YM}(A,B)$ are
\[
\begin{cases}
d_h^* F_A = 0,\\
d_{h_A}^* H = 0.
\end{cases}
\]
\end{enumerate}
\end{proposition}

\begin{proof}
\textit{(i)} The connection transforms by $A\mapsto A+u$ with $u$ closed and basic, so
$F_{g\cdot A}=F_A$ and
\[
  \dvol_{h_{g\cdot A}}=(A+u)\wedge\pi^*(\dvol_h)=A\wedge\pi^*(\dvol_h)=\dvol_{h_A},
\]
since $u\wedge\pi^*(\dvol_h)$ is a basic $3$-form on $\Sigma$. In particular the
metric $h_A$, and hence the Hodge star $\ast_{h_A}$, are unchanged. By
Lemma~\ref{lem:cocycle-repair}, the transformed data again patches to a $B$-field
on $Z$ with $\widetilde H=H$, thus both $F_A$ and $H$ are strictly invariant  for every level $\lambda$, and  therefore so is $\mathrm{YM}(A,B)$.
Writing $H=f\,\dvol_{h_A}$, one has
\[
\int_Z H\wedge *_{h_A}H
=
\int_Z f^2\,d\operatorname{vol}_{h_A}.
\]

\medskip

\textit{(ii)} 
Since $\Sigma$ is two-dimensional, any $i\mathbb R$-valued $2$-form is proportional to the volume form. Thus
\[
F_A=f_A\,d\operatorname{vol}_h,
\qquad f_A\in C^\infty(\Sigma;i\mathbb R).
\]
By Chern--Weil theory,
\[
\int_\Sigma F_A
=
2\pi i\,\langle c_1(Z),[\Sigma]\rangle,
\]
and hence
\[
\int_\Sigma f_A\,d\operatorname{vol}_h
=
2\pi i\,\langle c_1(Z),[\Sigma]\rangle.
\]

Let
\[
\bar f
:=
\frac{1}{\operatorname{Area}_h(\Sigma)}
\int_\Sigma f_A\,d\operatorname{vol}_h
=
\frac{2\pi i\,\langle c_1(Z),[\Sigma]\rangle}{\operatorname{Area}_h(\Sigma)}.
\]
Decomposing $f_A=(f_A-\bar f)+\bar f$ and integrating, the cross term vanishes, so
\[
\int_\Sigma |F_A|^2\,d\operatorname{vol}_h
=
\int_\Sigma |f_A-\bar f|^2\,d\operatorname{vol}_h
+
|\bar f|^2\,\operatorname{Area}_h(\Sigma).
\]
Substituting back yields exactly \eqref{22YM}.

\medskip

\textit{(iii)} 
We compute the first variation. Note that
\[
d\operatorname{vol}_{h_A}=A\wedge \pi^*(d\operatorname{vol}_h).
\]
For a variation $A_t=A+t\,\pi^*a$, one checks that
\[
d\operatorname{vol}_{h_{A_t}}=d\operatorname{vol}_{h_A},
\]
since $\pi^*a\wedge \pi^*(d\operatorname{vol}_h)=0$. In particular,  the term
\[
\int_Z H\wedge *_{h_A}H=\int_Z f^2\,d\operatorname{vol}_{h_A}
\]
is independent of $A$.

Now let $A_t=A+t\,\pi^*a$ with $a\in\Omega^1(\Sigma)$, so that $F_{A_t}=F_A+t\,da$, and hence
\[
\frac{d}{dt}\Big|_{t=0}
\int_\Sigma |F_{A_t}|^2\,d\operatorname{vol}_h
=
2\int_\Sigma \langle F_A,da\rangle\,d\operatorname{vol}_h
=
2\int_\Sigma \langle d_h^*F_A,a\rangle\,d\operatorname{vol}_h.
\]
Thus $d_h^*F_A=0$.

Similarly, for $B_t=B+t\,\beta$, one has $H_t=H+t\,d\beta$, and therefore
\[
\frac{d}{dt}\Big|_{t=0}
\int_Z H_t\wedge *_{h_A}H_t
=
2\int_Z \langle H,d\beta\rangle\,d\operatorname{vol}_{h_A}
=
2\int_Z \langle d_{h_A}^*H,\beta\rangle\,d\operatorname{vol}_{h_A}.
\]
Hence $d_{h_A}^*H=0$.

Combining the two equations gives the result.
\end{proof}

\subsection{Hamiltonian geometry on the configuration space}\label{sec:hamgeo}

On $\Ca$ define a $2$-form
\be
\label{eq:Omega-def}
\Omega :=\int_{\Sigma} \bigl[\;\delta A\wedge \delta A\;+\;\iota_v\delta B\wedge \iota_v\delta B\;\bigr],
\ee
with $v$ the vector field generating the $\bbT$-action along the fibres of $Z\to\Sigma$.
This $2$-form is translation invariant on each model tangent space.
It is also clear that $\Omega$ degenerates along the $b_1$-direction of
\eqref{eq:local-decomp}, since $\iota_vB_\alpha=b_{0,\alpha}$, so that
$\Omega$ sees only $\delta A$ and $\delta b_0$, and not $\delta b_1$.

The form $\Omega$ is  insensitive to the level $\lambda \in \ZZ$. Indeed, by \eqref{eq:killing-1} the twist
contributes $\lambda\,df\wedge\iota_vB$ to the Killing vector field of $f\in\mathfrak g_1$, and by
\eqref{eq:iota-twist} this contribution satisfies $\iota_v(\lambda\,df\wedge\iota_vB)=0$. It
therefore drops out of the second term of \eqref{eq:Omega-def} identically, while leaving the
first term untouched. {\em Consequently $\Omega$, its moment map, and the identity
$\mathrm{YM}=|\mu|^2$ below are the same for every $\lambda$, so the level is invisible to the
presymplectic geometry of $\Ca$}.

\medskip

\begin{proposition}\label{prop:moment}The action of $\GG_\lambda$ on $\Ca$ given by \eqref{gautran} is
Hamiltonian for the $2$-form \eqref{eq:Omega-def}, with moment map
\be\label{moment}
\mu(A,B)=\int_\Sigma F_A\,\cdot\;+\;\int_\Sigma \iota_vH\wedge(\iota_v\cdot)
\;\in\;\mathfrak g_1^*\oplus\mathfrak g_2^* ,
\ee
where $F_A=dA$ and $H=dB$. 
\end{proposition}

\begin{proof}
For any $f\in C^\infty(\Sigma,\RR)$, one has
\[
\int_{\Sigma} \delta F_A \cdot f
\;=\; \int_{\Sigma} d(\delta A)\, f
\;=\; \int_{\Sigma} df \wedge \delta A .
\]
For any $\omega\in  \Omega^1(Z)^{\bbT}$, one has
\[
\int_{\Sigma} \delta(\iota_v H)\wedge \iota_v\omega
\;=\; \int_{\Sigma} d(\iota_v \delta B)\wedge \iota_v\omega
\;=\; \int_{\Sigma} \iota_v(d\omega)\wedge \iota_v \delta B .
\]

These identities are precisely the variational statements that identify the contraction of a vector
field with $d\langle \mu,(\,f,\omega\,)\rangle$, yielding \eqref{moment}.
\end{proof}

\begin{proposition}\label{YMmoment}
With respect to the induced $L^2$-metric, one has
\[
\mathrm{YM}(A,B)=|\mu(A,B)|^2.
\]
\end{proposition}

\begin{proof}
We compute the  squared $L^2$-norm  of each component of the moment map.
For the curvature component, one has
\[
\left|\int_\Sigma F_A\,\cdot\right|^2
=
\int_\Sigma |F_A|^2\,d\operatorname{vol}_h.
\]
For the second component, write $H=f\,d\operatorname{vol}_{h_A}$
and since
\[
d\operatorname{vol}_{h_A}=A\wedge \pi^*(d\operatorname{vol}_h),
\quad \text{and} \quad A(v)=1,
\]
we obtain
\[
\iota_v(d\operatorname{vol}_{h_A})
=
\pi^*(d\operatorname{vol}_h),
\qquad
\iota_v H
=
f\,\pi^*(d\operatorname{vol}_h).
\]
Now for any $\omega\in \Omega^1(Z)^{\mathbb T}$,
\[
\int_\Sigma \iota_v H \wedge (\iota_v\omega)
=
\int_\Sigma f\,\pi^*(d\operatorname{vol}_h)\wedge \iota_v\omega
=
\int_Z f\,A\wedge \pi^*(d\operatorname{vol}_h)\wedge \omega
=
\int_Z f\,\omega\,d\operatorname{vol}_{h_A}.
\]
It follows that the $L^2$-norm of the second component is
\[
\left|\int_\Sigma \iota_v H \wedge (\iota_v\cdot)\right|^2
=
\int_Z f^2\,d\operatorname{vol}_{h_A}
=
\int_Z H\wedge *_{h_A}H.
\]
Combining the two components, we obtain
\[
|\mu(A,B)|^2
=
\int_\Sigma |F_A|^2\,d\operatorname{vol}_h
+
\int_Z H\wedge *_{h_A}H
=
\mathrm{YM}(A,B),
\]
as claimed.
\end{proof}


\section{Picard transformation and abelian gauge theory on free loop spaces} \label{gaugeloop}

In this section we explain how the action of the Picard group on
\(B\)-fields is reflected on the free loop space. For differential forms and the fundamental equivariant cohomology theory on
loop spaces and relation to index theory, see ~\cite{A85,B85}. Related developments involving T-duality and loop space
geometry may also be found in~\cite{FH08,JP, GJP,  HM15, HM18, HM22, HM24, HM25}.

Recall that a \(B\)-field on a smooth manifold \(M\) determines a holonomy line bundle
\[
\cL^B \longrightarrow LM
\]
equipped with a natural connection \(\nabla^{\cL^B}\).
Our goal is to identify the effect of the Picard action
\[
(L,\nabla^L)\in \pic(M)
\]
on the corresponding loop space data.
More precisely, we show that replacing \(B\) by
\[
(L,\nabla^L)\cdot B
\]
does not change the underlying holonomy line bundle, but modifies the
connection \(\nabla^{\cL^B}\) by the gauge transformation
\[
\mathrm{hol}(\nabla^L):LM\to\TT,
\]
where \(\mathrm{hol}(\nabla^L)\) denotes the holonomy of
\(\nabla^L\) around loops. Thus the Picard action on gerbes is
transgressed to the ordinary gauge action on the loop space.

Throughout this subsection, \(M\) is an arbitrary smooth manifold.
In the T-duality applications considered later, we specialize to
\(M=Z\) and assume all geometric data are \(\TT\)-invariant.

\medskip

Let $M$ be a smooth manifold. We recall several standard constructions on the loop space $LM$.  
The evaluation map
\[
ev : S^{1}\times LM \longrightarrow M, \qquad (t,\gamma)\mapsto \gamma(t),
\]
induces the \emph{transgression map}
\[
\tau : \Omega^{\bullet}(M)\to \Omega^{\bullet-1}(LM),\qquad 
\tau(\xi)=\int_{S^{1}} ev^{*}\xi .
\]
Given $\omega\in \Omega^{i}(M)$, define for each $s\in [0,1]$ a form $\widehat\omega_{s}\in\Omega^{i}(LM)$ by
\[
\widehat\omega_s(X_1,\ldots,X_i)(\gamma)
=
\omega\bigl(X_1|_{\gamma(s)},\ldots,X_i|_{\gamma(s)}\bigr),
\]
for vector fields $X_1,\ldots,X_i$ defined near the loop $\gamma$.  
Since $d\widehat\omega_s=\widehat{d\omega}_s$, the averaged form
\[
\overline\omega := \int_{0}^{1}\widehat\omega_s\,ds \in \Omega^{i}(LM)
\]
is invariant under the rotation action of $S^{1}$ generated by the vector field~$K$.  
Moreover,
\[
\tau(\omega)=\iota_{K}\,\overline\omega .
\]

\medskip

Let $\{\sU_\alpha\}$ be an open cover of $M$.  
To ensure that $\{L\sU_\alpha\}$ covers $LM$, we choose a \emph{Brylinski open cover}: a maximal family satisfying
\[
H^{2}(\sU_{\alpha_I})=H^{3}(\sU_{\alpha_I})=0
\quad \text{for all finite intersections }\sU_{\alpha_I}.
\]
Such covers always exist and include, for example, tubular neighbourhoods of individual loops.

Let $H$ be a closed $3$-form on $M$ with integral periods.  
Since $H^{3}(\sU_\alpha)=0$, we may write
\[
H|_{\sU_\alpha}=dB_\alpha,
\]
for purely imaginary $2$-forms $B_\alpha$, and since $H^{2}(\sU_{\alpha\beta})=0$,
\[
B_\beta - B_\alpha = dA_{\alpha\beta},
\]
for purely imaginary $1$-forms $A_{\alpha\beta}$.  
The triple $(H,B,A)$ defines a connective structure of a gerbe $\cG_B$ on~$M$.

A geometric realization of $\cG_B$ is a family of line bundles  
\[
(L_{\alpha\beta},\nabla^{L}_{\alpha\beta})\longrightarrow \sU_{\alpha\beta},
\]
with the usual gerbe compatibility conditions.  
Since $H^{2}(\sU_{\alpha\beta})=0$, each $L_{\alpha\beta}$ is trivial, and
\[
\nabla^L_{\alpha\beta}=d+A_{\alpha\beta}, \qquad 
(\nabla^L_{\alpha\beta})^{2}=B_\beta - B_\alpha .
\]

\medskip

The holonomy of $\cG_B$ defines a line bundle
\[
\cL^{B} \longrightarrow LM,
\]
equipped with Brylinski local sections $\{\sigma_\alpha\}$ over $\{L\sU_\alpha\}$.  
On overlaps these satisfy
\[
\sigma_\alpha
=
e^{-\int_{0}^{1}\iota_K A_{\alpha\beta}}
\sigma_\beta
=
e^{-\tau(A_{\alpha\beta})}\sigma_\beta .
\]
The natural connection on $\cL^{B}$ is described in this basis by
\be\label{connection}
\nabla^{\cL^{B}} 
= d - \iota_K \overline B_\alpha
= d - \tau(B_\alpha),
\ee
with curvature
\[
F_B = (\nabla^{\cL^{B}})^2 = -\tau(H).
\]

\begin{remark}\label{add2form}
From the construction, we see that adding a global 2-form \(\omega\) to the $B$-field does not change the holonomy line bundle \(\cL^{B}\). 
It only modifies the connection \(\nabla^{\cL^{B}}\) by the transgressed 1-form \(\tau(\omega)\).
\end{remark}

\medskip

Let $\pic(M)$ denote the Picard group of line bundles with connection on $M$.  
The action $(L,\nabla^L)\cdot B$ on $B$-fields has the following interpretation on the loop space:

\begin{proposition}\label{loop gauge}
For $(L,\nabla^L)\in \pic(M)$,
\[
\nabla^{\cL^{(L,\nabla^L)\cdot B}}
=
\nabla^{\cL^{B}} + d\ln \mathrm{hol}(\nabla^L),
\]
where $\mathrm{hol}(\nabla^L):LM\to \TT$ is the holonomy of the connection $\nabla^L$ around loops.
\end{proposition}

\begin{proof}
By Proposition~6.1.1 of~\cite{Bry},
\[
-\iota_K\,\overline{(\nabla^L)^2}
=
\tau\bigl((\nabla^L)^2\bigr)
=
d\ln \mathrm{hol}(\nabla^L).
\]
Using \eqref{connection}, we compute under the basis $\{\sigma_\alpha\}$:
\[
\nabla^{\cL^{(L,\nabla^L)\cdot B}}
=
d - \iota_K\bigl(\overline B_\alpha + \overline{(\nabla^L)^2}\bigr)
=
d - \iota_K \overline B_\alpha + d\ln \mathrm{hol}(\nabla^L)
=
\nabla^{\cL^{B}} + d\ln \mathrm{hol}(\nabla^L).
\]
The desired equality follows. 
\end{proof}

Thus the Picard transformation $(L,\nabla^L)\cdot B$ induces the gauge transformation
\[
\mathrm{hol}(\nabla^L):LM\to\TT
\]
on the connection $\nabla^{\cL^{B}}$ of the holonomy line bundle $\cL^{B}\to LM$.


\section{Moduli spaces of flat mixed fields} \label{moduli}

The aim of this section is to study the presymplectic quotient of the configuration space of
mixed fields by the gauge group. We shall see that, although the defining equations are linear, the
quotient exhibits a non-trivial geometry. 

We begin by introducing the locus on which the  reduction
takes place. Denote by $\Caf$ the subset of the configuration space $\Ca$ given by
\[
\Caf \;:=\; \bigl\{\, (A,B)\ \big|\ F_A\ \text{exact on }\Sigma,  \,  H \ \text{exact on }Z\bigr\}\subset \Ca.
\]
It follows immediately from \eqref{gautran} that $\Caf$ is preserved by the action of the gauge group $\GG$.

Associated to $\Caf$ is a topologically trivial complex line bundle $\cL$ equipped with a connection
\[
\nabla^{\cL} \;=\; d + \Lambda,
\]
such that at a point $(A,B)$ and for a tangent vector
\[
(\delta A,\delta B)\in T_{(A,B)}\Caf
\cong
\Omega^1(\Sigma)\times\Omega^2(Z)^{\TT},
\]
the connection $1$-form is defined by
\be\label{prequantumconnection}
\Lambda_{(A,B)}\bigl(\delta A,\delta B\bigr)
\;=\;
\int_{\Sigma}\Bigl(\,A^{\mathrm{bas}}\wedge \delta A
\;+\;
\iota_v B \wedge \iota_v \delta B\,\Bigr),
\ee
where $A^{\mathrm{bas}}$ denotes the basic form on $\Sigma$ associated to $A$, and $v$ is the vector field 
generating the $\TT$-action along the fibres of $Z\to\Sigma$.

A direct computation shows that the curvature of $\nabla^{\cL}$ is precisely
\[
(\nabla^{\cL})^2 \;=\; \Omega,
\]
where $\Omega$ is the $2$-form on the configuration space introduced in \eqref{eq:Omega-def}. 
Consequently, $(\cL,\nabla^{\cL})$ is a prequantum line bundle ~\cite{Kostant} over $\Caf$:
\[
\begin{CD}
(\cL,\nabla^{\cL})
\\
@V\pi VV
\\
\Caf .
\end{CD}
\]

The action of $\GG_\lambda$ on $\Caf$ lifts naturally to an action on $\cL$, and the connection $\nabla^{\cL}$ is 
invariant under this lifted action. Moreover, one verifies that
\[
\mu(X)
=
L_{\rho(X)}
-
\nabla^{\cL}_{\rho(X)},
\]
where $\mu$ is the moment map of Proposition~\ref{prop:moment}, $X\in \mathfrak g_1^*\oplus \mathfrak g_2^*$, 
and $\rho(X)$ denotes the corresponding Killing vector field on $\Caf$.
We therefore obtain the presymplectic reduction
\[
\Caf/ \GG_\lambda
=
\mu^{-1}(0)/\GG_\lambda,
\]
which we shall refer to as the {\em moduli space of flat mixed fields}.  The  quotient has the following simple finite-dimensional description.

\begin{lemma}\label{lem:finite-model}
Let $\lambda\in\ZZ$ be the level of the action \eqref{gautran}. The product
\( \RR^{2g} \times \RR^{2g} \times \RR/\ZZ \) carries a natural
\( (\mathbb{Z}^{2g}\times \mathbb{Z}^{2g}) \)-action defined by
\begin{equation} \label{discreteaction}
 (\mathbf{m}, \mathbf{n})  \cdot (\mathbf{x}, \mathbf{y}, [t]) := \left( \mathbf{x}+\mathbf{m}, \mathbf{y}+\mathbf{n},\ 
[t + \phi^\lambda_{(\mathbf m,\mathbf n)}(\mathbf x,\mathbf y)]\right),
\end{equation}
\[\phi^\lambda_{(\mathbf m,\mathbf n)}(\mathbf x,\mathbf y):=
\mathbf{n}\Theta_g \mathbf{x}^T+(\lambda-1)\,\mathbf{m}\Theta_g \mathbf{y}^T,
\]
where 
$$ (\mathbf{m}, \mathbf{n}) = (m_1, \cdots, m_{2g},  n_1, \cdots,  n_{2g}) \in  \mathbb{Z}^{2g}\times \mathbb{Z}^{2g} $$ acts on 
$$ (\mathbf{x}, \mathbf{y}, [t])=(x_1, \cdots, x_{2g},  y_1, \cdots,  y_{2g}, [t])\in \RR^{2g} \times \RR^{2g} \times \RR/\ZZ. $$ 
There is a one-to-one correspondence between the orbit space of the action of $\GG_\lambda$ on $\mu^{-1}(0)$ and the orbit space of 
the $ (\mathbb{Z}^{2g}\times \mathbb{Z}^{2g})$-action on \( \RR^{2g} \times \RR^{2g} \times \RR/\ZZ \) described in \eqref{discreteaction}. 
\end{lemma}
\begin{proof} 
Take an element \((A,B)\in \mu^{-1}(0)\). Fix a trivialization of \(Z\), and let \(\theta\) denote the resulting global vertical coordinate. Then we may write
\begin{equation}\label{expression}
\left\{
\begin{array}{rl}
A &= a_0+d\theta,\\[2pt]
B &= b_1 + A\wedge b_0,
\end{array}
\right.
\end{equation}
where \(a_0,b_0\in \Omega^{1}(\Sigma)\) and \(b_1\in \Omega^{2}(\Sigma)\). Note that, since
$\iota_vA=1$ and $b_0,b_1$ are basic,
\be\label{eq:iota-B}
\iota_vB=(\iota_vA)\,b_0-A\wedge(\iota_vb_0)=b_0 .
\ee

In this description the datum \((A,B)\) is encoded by the triple \((a_0,b_0,b_1)\). 
Moreover, the condition \((A,B)\in \mu^{-1}(0)\) forces \(a_0\) and \(b_0\) to be closed:
\[
a_0\in \Omega^{1}_{\mathrm{cl}}(\Sigma),\qquad b_0\in \Omega^{1}_{\mathrm{cl}}(\Sigma).
\]
Indeed $F_A=da_0$, while
$H=db_1+da_0\wedge b_0-a_0\wedge db_0-d\theta\wedge db_0$, whose first three terms are basic
forms of degree $3$ on the surface $\Sigma$ and hence vanish. Thus, $H=0$ is equivalent to
$db_0=0$, and $b_1\in\Omega^2(\Sigma)$ is unconstrained.

Now choose an element \((g,(L,\nabla^{L}))\in \GG_\lambda\), and recall the notation $u=g^{-1}dg$,
$\eta=(\nabla^{L})^{2}$. By \eqref{gautran} and \eqref{eq:iota-B}, the transformed pair
\((\widetilde A,\widetilde B)\) is given by
\[
\widetilde A = A + u,\qquad \widetilde B = B+\eta+\lambda\,u\wedge b_0 .
\]
By \eqref{eq:iota-twist}, $\iota_v(u\wedge b_0)=0$, hence
\[
\widetilde b_0=\iota_v\widetilde B=b_0+\iota_v\eta .
\]
Expressing \((\widetilde A,\widetilde B)\) again in the form \eqref{expression} and subtracting,
\[
\widetilde b_1=\widetilde B-\widetilde A\wedge\widetilde b_0
=\bigl(b_1+A\wedge b_0+\eta+\lambda\,u\wedge b_0\bigr)
-(A+u)\wedge(b_0+\iota_v\eta),
\]
so that the triple \((a_0,b_0,b_1)\) is carried to \((\widetilde a_0,\widetilde b_0,\widetilde b_1)\) with
\begin{equation}\label{afterchange}
\left\{
\begin{array}{rl}
\widetilde a_0 &= a_0 + u,\\[4pt]
\widetilde b_0 &= b_0 + \iota_v \eta,\\[4pt]
\widetilde b_1 &= b_1
+ \bigl(\eta - d\theta\wedge \iota_v\eta - u\wedge \iota_v\eta\bigr)
+ (\iota_v\eta)\wedge a_0
+(\lambda-1)\, u\wedge b_0 .
\end{array}
\right.
\end{equation}
Here we used $A\wedge\iota_v\eta=d\theta\wedge\iota_v\eta+a_0\wedge\iota_v\eta$ and
$-a_0\wedge\iota_v\eta=(\iota_v\eta)\wedge a_0$, and the $\lambda$-dependence enters solely
through the cancellation of $\lambda\,u\wedge b_0$ against the term $-u\wedge b_0$ coming from
$-\widetilde A\wedge\widetilde b_0$. 
By \eqref{eq:iota-eta}, \(\iota_v\eta\) is a closed \(1\)-form with integral periods. Likewise \(u\) has integral periods, and the \(2\)-form
\[
\eta - d\theta\wedge \iota_v\eta - u\wedge \iota_v\eta
\]
descends to \(\Sigma\) and represents an integral cohomology class.

Finally, fix a symplectic basis \(\{\theta_1,\dots,\theta_g,\eta_1,\dots,\eta_g\}\) of \(H^{1}(\Sigma;\RR)\), normalised so that
$\int_\Sigma\alpha\wedge\beta=[\alpha]\,\Theta_g\,[\beta]^T$. Write the cohomology classes as
\be \label{expression2}
[a_0]=\sum_{i=1}^{g}\bigl(x_i\theta_i + x_{i+g}\eta_i\bigr),\qquad
[b_0]=\sum_{i=1}^{g}\bigl(y_i\theta_i + y_{i+g}\eta_i\bigr),
\ee
and similarly
\be \label{expression1}
[u]=\sum_{i=1}^{g}\bigl(m_i\theta_i + m_{i+g}\eta_i\bigr),\qquad
[\iota_v\eta]=\sum_{i=1}^{g}\bigl(n_i\theta_i + n_{i+g}\eta_i\bigr),
\ee
and set $[t]:=[\int_\Sigma b_1]\in\RR/\ZZ$.

The subgroup $\GG^0<\GG_\lambda$ of elements with $[u]=0$ and $[\iota_v\eta]=0$ acts on
$(a_0,b_0,b_1)$ by shifting $a_0$ and $b_0$ by exact $1$-forms and $b_1$ by an exact $2$-form
together with an integral class. Note that for such elements $u=df_1$ is exact and, on the flat
locus, $\lambda\,u\wedge b_0=\lambda\,d(f_1b_0)$ is exact as well, so that $\GG^0$ is independent
of $\lambda$. Hence $\GG^0$ acts trivially on $(\mathbf x,\mathbf y,[t])$, and the residual action
is that of $\GG_\lambda/\GG^0\cong\ZZ^{2g}\times\ZZ^{2g}$. Passing to \eqref{afterchange} and discarding
the integral classes which vanish in $\RR/\ZZ$, these being
$[\eta-d\theta\wedge\iota_v\eta]$ and $[u\wedge\iota_v\eta]$, the latter with periods
$\mathbf m\Theta_g\mathbf n^T\in\ZZ$, the surviving shift of the fibre coordinate is
\[
\Delta t=\int_\Sigma\Bigl((\iota_v\eta)\wedge a_0+(\lambda-1)\,u\wedge b_0\Bigr)
=\mathbf n\Theta_g\mathbf x^T+(\lambda-1)\,\mathbf m\Theta_g\mathbf y^T
=\phi^\lambda_{(\mathbf m,\mathbf n)}(\mathbf x,\mathbf y).
\]
That \eqref{discreteaction} is indeed an action is verified directly: composing
$(\mathbf m,\mathbf n)$ and then $(\mathbf m',\mathbf n')$ shifts the fibre by
$\phi^\lambda_{(\mathbf m,\mathbf n)}(\mathbf x,\mathbf y)
+\phi^\lambda_{(\mathbf m',\mathbf n')}(\mathbf x+\mathbf m,\mathbf y+\mathbf n)$, which differs from
$\phi^\lambda_{(\mathbf m+\mathbf m',\mathbf n+\mathbf n')}(\mathbf x,\mathbf y)$ by
$\mathbf n'\Theta_g\mathbf m^T+(\lambda-1)\mathbf m'\Theta_g\mathbf n^T\in\ZZ$, using the integrality of $\Theta_g$.
\end{proof}

\begin{theorem} \label{main1}
Let $\lambda\in\ZZ$ and let $\Theta_g$ be the
standard symplectic matrix
\[
\Theta_g =
\begin{pmatrix}
0 & I_g \\
- I_g & 0
\end{pmatrix}
\in Sp(2g,\RR).
\]
\begin{enumerate}
\item[(i)] For $\lambda\ne0$, the moduli space $\mu^{-1}(0)/\GG_\lambda$ is a
 Heisenberg contact manifold of dimension $4g+1$, that is a principal circle bundle
over the torus $T^{2g}\times T^{2g}$ with first Chern class
\[
c_1(\mu^{-1}(0)/\GG_\lambda)
=
-\,\lambda\sum_{j,k} (\Theta_g)_{jk} \, \alpha_j \cup \beta_k,
\]
where $\alpha_j \in H^1(T^{2g}, \mathbb{Z})$ is the dual basis to the coordinate $x_j$ and
$\beta_k \in H^1(T^{2g}, \mathbb{Z})$ is the dual basis to the coordinate $y_k$.

\item[(ii)] For $\lambda=0$, the moduli space is a trivial circle bundle over $T^{2g}\times T^{2g}$, 
\[
\mu^{-1}(0)/\GG_\lambda\ \cong\ T^{2g}\times T^{2g}\times \RR/\ZZ\ =\ T^{4g+1}.
\]
\end{enumerate}
\end{theorem}

\begin{proof}

\noindent
{\em (i)}
By Lemma~\ref{lem:finite-model} the moduli space is the quotient of
$\mathbb R^{2g}\times\mathbb R^{2g}\times\mathbb R/\mathbb Z$ by the action
\eqref{discreteaction}. That action is free, since
$(\mathbf m,\mathbf n)\cdot(\mathbf x,\mathbf y,[t])=(\mathbf x,\mathbf y,[t])$ already forces
$\mathbf m=\mathbf n=\mathbf 0$ from the first two components, and it is properly discontinuous,
hence
\[
\mu^{-1}(0)/\GG_\lambda= 
(\mathbb R^{2g}\times \mathbb R^{2g}\times \mathbb R/\mathbb Z)
/(\mathbb Z^{2g}\times \mathbb Z^{2g})
\] 
is a smooth principal circle bundle over  $T^{2g}\times T^{2g}$.

To determine its Chern class we construct an invariant connection on
$\mathbb R^{2g}\times \mathbb R^{2g} \times \mathbb R/\mathbb Z $. Let
\[
\alpha
=
dt+(1-\lambda)\,\mathbf x\Theta_g d\mathbf y^{T}
-\mathbf y\Theta_g d\mathbf x^{T}.
\]
Since $(\mathbf m,\mathbf n)^*dt=dt+d\phi^\lambda_{(\mathbf m,\mathbf n)}$ with
\[
d\phi^\lambda_{(\mathbf m,\mathbf n)}
=
\mathbf n\Theta_g d\mathbf x^{T}
+(\lambda-1)\,
\mathbf m\Theta_g d\mathbf y^{T},
\]
one computes
\[
(\mathbf m,\mathbf n)^*\alpha
= dt+\mathbf n\Theta_g d\mathbf x^T+(\lambda-1)\mathbf m\Theta_g d\mathbf y^T
+(1-\lambda)(\mathbf x+\mathbf m)\Theta_g d\mathbf y^T
-(\mathbf y+\mathbf n)\Theta_g d\mathbf x^T
=\alpha ,
\]
the terms in $\mathbf m$ and $\mathbf n$ cancelling in pairs. Hence \(\alpha\) descends to a
connection \(1\)-form on \(\mu^{-1}(0)/\GG_\lambda\). Its curvature is, using $\Theta_g^T=-\Theta_g$,
\[
d\bigl(\mathbf y\Theta_gd\mathbf x^T\bigr)
=\sum_{j,k}(\Theta_g)_{jk}\,dy_j\wedge dx_k
=-\sum_{j,k}(\Theta_g)_{jk}\,dx_k\wedge dy_j
= d\mathbf x \Theta_g d\mathbf y^{T},
\]
hence
\[
d\alpha
=
(1-\lambda)\,d\mathbf x \Theta_g d\mathbf y^{T}-d\mathbf x \Theta_g d\mathbf y^{T}
=
- \lambda\,
\sum_{j,k}
(\Theta_g)_{jk}\,
dx_j\wedge dy_k .
\]
Therefore $c_1(\mu^{-1}(0)/\GG_\lambda)=[d\alpha]$, the fibre coordinate $t\in\RR/\ZZ$ having period $1$.
With $\alpha_j,\beta_k\in H^1(T^{2g},\mathbb Z)$ dual to \(x_j\) and \(y_k\),
\[
c_1(\mu^{-1}(0)/\GG_\lambda)
=
-\lambda
\sum_{j,k}
(\Theta_g)_{jk}
\,
\alpha_j\cup\beta_k = 
-\lambda
\sum_{i=1}^{g}
\bigl(
\alpha_i\cup\beta_{i+g}
-
\alpha_{i+g}\cup\beta_i
\bigr).
\]
Introducing a symplectic basis
\[
e_i=\alpha_i,
\qquad
f_i=\beta_{i+g},
\qquad
(1\le i\le g),
\]
and
\[
e_{g+i}=-\alpha_{g+i},
\qquad
f_{g+i}=\beta_i,
\qquad
(1\le i\le g),
\]
one has
\[
c_1(\mu^{-1}(0)/\GG_\lambda)
=
-\lambda
\sum_{j=1}^{2g}
e_j\wedge f_j .
\]
For $\lambda\neq0$ the $2$-form $d\alpha$ is nondegenerate on $\RR^{4g}$, so
$\alpha\wedge(d\alpha)^{2g}\neq0$ and \(\mu^{-1}(0)/\GG_\lambda\) is a  Heisenberg contact
manifold of dimension \(4g+1\), with contact 1-form $\alpha$. 

\medskip
\noindent
{\em (ii)} For $\lambda=0$, the Chern class $c_1(\mu^{-1}(0)/\GG_\lambda)$ is zero and the circle bundle trivialises,
\[
\mu^{-1}(0)/\GG_\lambda\ \cong\ T^{2g}\times T^{2g}\times \RR/\ZZ\ =\ T^{4g+1} .
\]
In particular $\mu^{-1}(0)/\GG_\lambda$ carries no contact structure.
\end{proof}

\subsection{The topology of the moduli space}\label{sec:Ktheory}
Throughout this
subsection $\lambda\in\ZZ\setminus\{0\}$ is fixed and $M_{|\lambda|}\xrightarrow{\ \pi\ }T^{4g}$
denotes the principal circle bundle underlying $\mu^{-1}(0)/\GG_\lambda$ at level $\lambda$, with Euler
class
\[
c_1(M_{|\lambda|})= \pm \lambda\,\omega\in H^2(T^{4g};\Z),
\qquad
\omega=\sum_{j=1}^{2g} e_j\wedge f_j ,
\]
We write $V:=H^1(T^{4g};\Z)\cong\Z^{4g}$, so that
$H^r(T^{4g};\Z)\cong\Lambda^r V$ is free of rank $\binom{4g}{r}$, and introduce the integral
Lefschetz operator
\[
L_r:\Lambda^r V\longrightarrow \Lambda^{r+2}V,\qquad L_r(\xi)=\omega\wedge\xi,
\]
with $L=\bigoplus_r L_r$. Note that the sign of the Euler class is immaterial for the discussion below, since
$\coker(-\lambda L_r)=\coker(\lambda L_r)$ and $\ker(-\lambda L_r)=\ker(\lambda L_r)$.

The fundamental group of $M_{|\lambda|}$ is the integer Heisenberg lattice
$H(|\lambda|,\dots,|\lambda|)$ of \cite{LP99}, namely a central extension of
$\ZZ^{2g}\times\ZZ^{2g}$ by $\ZZ$ twisted by the 2-cocycle
\be\label{eq:cocycle}
c_\lambda((\mathbf m,\mathbf n),(\mathbf m',\mathbf n'))
=
\mathbf n'\Theta_g\mathbf m^{T}
+(\lambda-1)\,
\mathbf m'\Theta_g\mathbf n^{T}.
\ee

The integral cohomology of $M_{|\lambda|}$ was computed by Lee and Packer in \cite{LP}.

\begin{theorem}[{\cite[Thm.~2.1]{LP}}]\label{thm:cohomology}
Let $\lambda\in\ZZ\setminus\{0\}$ and $k \in \mathbb N$, then
\[
H^k(M_{|\lambda|};\Z)\ \cong\ \coker\big(\lambda L_{k-2}\big)\ \oplus\ \ker \big(\lambda L_{k-1}\big)
\ \cong\ \Z^{\,b_k}\ \oplus\ \bigoplus_{j\ge 1}\Z_{|\lambda|j}^{\,\varepsilon_{k-2}(j)},
\]
with   rank
\[
b_k=
\begin{cases}
\binom{4g}{k}-\binom{4g}{k-2}, & 0\le k\le 2g,\\[2mm]
\binom{4g}{k-1}-\binom{4g}{k+1}, & 2g+1\le k\le 4g+1,
\end{cases}
\]
and torsion multiplicities 
\be\label{eq:eps-closed}
\varepsilon_r(j)\ =\ \binom{4g}{\,r+2-2j\,}-\binom{4g}{\,r-2j\,}
\ee
for $0\le r\le 2g-1$, and extended to $2g\le r\le 4g-2$ by $\varepsilon_r(j)=\varepsilon_{4g-2-r}(j)$.
\end{theorem}

The $K$-theory of $M_{|\lambda|}$ has been computed by Aslaksen--Lee--Packer
\cite{ALP}, but only for $\lambda =1$. Their strategy is to run the $K$-theoretic Gysin sequence of the circle bundle
$M_{|\lambda|}\to T^{4g}$ in parallel with its cohomological counterpart. In the cohomological Gysin sequence
the connecting map is the Lefschetz operator $\lambda L$, and passing  to
$K$-theory replaces the integral Euler class $\lambda\omega$ by the $K$-theoretic Euler class
$\Theta_\lambda:=1-e^{\lambda\omega}$. Since $K^\bullet(T^{4g})$ is torsion-free and the  Chern character
restricts to an isomorphism of integral lattices
$K^\bullet(T^{4g})\xrightarrow{\ \cong\ }H^\bullet(T^{4g};\Z)\cong\Lambda^\bullet V$, the entire
computation takes place in $\Lambda^\bullet V$ and the $K$-groups of $M_{|\lambda|}$ are read off
from the kernel and
cokernel of multiplication by $\Theta_\lambda$. The key point is to compare the
connecting map $\Theta_\lambda$ with the cohomological one $\lambda L$
{\em over} $\Z$, then $\coker(\lambda L)$, whose torsion is the cohomology torsion of
Theorem~\ref{thm:cohomology},  controls the torsion of the $K$-groups as well.

Aslaksen--Lee--Packer establish such an integral comparison at $\lambda=1$ by decomposing
$\Lambda^\bullet V$ into subspaces on which the Lefschetz operator acts, and exhibiting on each
block a unimodular change of basis conjugating the $K$-theoretic Euler class $e^{\omega}-1$ to
the Lefschetz operator $L$ \cite[Prop.~2.3, Prop.~2.6]{ALP}. The construction proceeds by
induction over the blocks and relies on an operator identity of the form $[X,L]=L$ that is 
homogeneous of degree one in $L$ \cite[Lem.~2.5]{ALP}. This identity is scale-invariant and for
that reason cannot relate $\lambda L$ to $e^{\lambda\omega}-1$ when $\lambda\ne1$. Namely, since
$e^{\lambda\omega}-1=\lambda L+\tfrac{\lambda^2}{2}L^2+\cdots$, the ratio of the $L^2$-coefficient
to the $L$-coefficient is $\lambda/2$, which varies with the level, whereas a conjugation built
from $[X,L]=L$ can only reproduce the fixed value $\tfrac12$ of the case $\lambda=1$.

To remove this restriction, we construct a single $\Z$-linear automorphism $\Psi_\lambda$ of the whole
exterior algebra $\Lambda^\bullet V$ intertwining $\lambda L$ with $e^{\lambda\omega}-1$. In place of the blockwise
conjugation at $\lambda=1$ \cite[Prop.~2.6]{ALP}, which it recovers in an equivalent form, one
obtains a single intertwiner valid at every level. 

\begin{lemma}\label{lem:conj}
Let $\lambda\in\ZZ$.  There is a $\Z$-linear automorphism
$\Psi=\Psi_\lambda$ of $\Lambda^\bullet V$ that preserves $\Lambda^{\mathrm{ev}}V$ and
$\Lambda^{\mathrm{odd}}V$, preserves the degree filtration and induces the identity on its
associated graded, and satisfies
\[
\Psi\circ(\lambda L)=\bigl(e^{\lambda\omega}-1\bigr)\circ\Psi .
\]
\end{lemma}

\begin{proof} 
For $1\le i\le 2g$ write $x_i:=e_i\wedge f_i\in\Lambda^2V$,
so that $\omega=\sum_{i=1}^{2g}x_i$. Because each $x_i$ is a wedge of two
distinct generators, the $x_i$ commute and
satisfy $x_i^2=0$. Let
\[
  R:=\mathbb Z[x_1,\dots,x_{2g}]\big/\bigl(x_1^2,\dots,x_{2g}^2\bigr)
\]
act on $\Lambda^\bullet V$ by wedge multiplication. The operator $\lambda L$
is precisely multiplication by the element
$\lambda\omega=\lambda\sum_i x_i\in R$, and since the $x_i$ square to zero the
exponential is a finite product,
\begin{equation}\label{eq:exp}
  e^{\lambda\omega}=\prod_{i=1}^{2g}e^{\lambda x_i}
   =\prod_{i=1}^{2g}\bigl(1+\lambda x_i\bigr)\in R,
\end{equation}
the middle equality using that the $x_i$ commute and
$e^{\lambda x_i}=1+\lambda x_i$ because $x_i^2=0$. Thus the whole comparison
of $\lambda L$ with $e^{\lambda\omega}-1$ takes place inside the commutative
square-zero algebra $R$, and reduces to producing a ring automorphism of $R$
carrying $\lambda\omega$ to $e^{\lambda\omega}-1$.

The ring $R$ is free over $\mathbb Z$ on the square-free monomials
$x_C:=\prod_{i\in C}x_i$, $C\subseteq\{1,\dots,2g\}$. Define
\[
  \psi(x_i):=x_i\prod_{j<i}\bigl(1+\lambda x_j\bigr).
\]
Since $x_i\cdot\psi(x_i)=x_i^2\prod_{j<i}(1+\lambda x_j)=0$, the elements
$\psi(x_i)$ again satisfy $\psi(x_i)^2=0$, so by the universal property of
$R$ this assignment extends to a ring endomorphism $\psi\colon R\to R$.

For every $1\le m\le 2g$ we have the elementary factorisation, 
\begin{equation}\label{eq:peel}
  \prod_{i\le m}\bigl(1+\lambda x_i\bigr)-1
  =\Bigl[\prod_{i\le m-1}\bigl(1+\lambda x_i\bigr)-1\Bigr]
   +\lambda x_m\prod_{j\le m-1}\bigl(1+\lambda x_j\bigr).
\end{equation}

Set $P_m:=\prod_{i\le m}\bigl(1+\lambda x_i\bigr)-1$, with the convention $P_0=\prod_{\varnothing}-1=0$. Then \eqref{eq:peel} reads
$P_m=P_{m-1}+\lambda\,\psi(x_m)$, and because $\psi$ is $\mathbb Z$-linear and a ring homomorphism, we get
\be \label{eq:conj}
\psi\bigl(\lambda\sum_i x_i\bigr)=\sum_i\lambda\,\psi(x_i)
   =\sum_{m=1}^{2g}\bigl(P_m-P_{m-1}\bigr)
   =P_{2g}-P_0
   =\prod_{i=1}^{2g}\bigl(1+\lambda x_i\bigr)-1 = e^{\lambda\omega}-1 .
\ee
Grade $R$ by monomial length. From the definition,
$\psi(x_C)=x_C+(\text{strictly longer monomials})$, so in the monomial basis
$\psi=\mathrm{id}+N$ with $N$ strictly length-raising, hence nilpotent.
Therefore $\psi$ is unitriangular and hence a $\mathbb Z$-linear
automorphism of $R$.

We can lift $\psi$ uniformly from $R$ to $\Lambda^\bullet V$ via the pair-free
decomposition \cite[\S1]{LP}. The monomial basis of $\Lambda^\bullet V$ consists, up to sign, of the
elements $\mu_S\wedge x_C$, where $\mu_S$ is a pair-free monomial (at
most one factor from each pair $\{e_i,f_i\}$),
$\operatorname{supp}S\subseteq\{1,\dots,2g\}$ records the occupied pairs, and
$C\subseteq\{1,\dots,2g\}\setminus\operatorname{supp}S$. Since
$x_i\wedge\mu_S=0$ for $i\in\operatorname{supp}S$, one
has an $R$-module decomposition
\[
  \Lambda^\bullet V=\bigoplus_S \mu_S\wedge R,
  \qquad
  \mu_S\wedge R\cong R/I_S,\quad I_S:=\bigl(x_i:i\in\operatorname{supp}S\bigr),
\]
the isomorphism being $r+I_S\mapsto\mu_S\wedge r$. Under $L=\omega\wedge(-)$
each block $R/I_S$ is carried to itself by multiplication by $\omega$.
For $i\in\operatorname{supp}S$ we have $\psi(x_i)\in(x_i)\subseteq I_S$, so
$\psi$ preserves $I_S$ and induces a unitriangular ring automorphism
$\psi_S$ of $R/I_S$. Define
\[
  \Psi\bigl(\mu_S\wedge r\bigr):=\mu_S\wedge\psi_S(r),
\]
extended additively over the blocks. Being blockwise unitriangular, $\Psi$
is a $\mathbb Z$-linear automorphism of $\Lambda^\bullet V$. Moreover, for
$h\in R$,
\begin{equation}\label{eq:intertwine}
  \Psi\bigl(h\cdot(\mu_S\wedge r)\bigr)
   =\mu_S\wedge\psi_S(hr)
   =\mu_S\wedge\psi_S(h)\,\psi_S(r)
   =\psi(h)\cdot\Psi\bigl(\mu_S\wedge r\bigr),
\end{equation}
where the middle equality uses that $\psi_S$ is a ring homomorphism and the
last uses $\psi_S(h)=\psi(h)\bmod I_S$.
Since $\psi_S(r)-r$ raises monomial length and multiplication by $x_C$
raises the exterior degree by $2|C|$, the operator $\Psi-\mathrm{id}$
strictly raises the exterior degree by even amounts. Consequently
$\Psi$ preserves the parity splitting
$\Lambda^\bullet V=\Lambda^{\mathrm{ev}}V\oplus\Lambda^{\mathrm{odd}}V$,
preserves the degree filtration, and induces the identity on the associated
graded. Finally, combining \eqref{eq:intertwine} (with
$h=\lambda\sum_i x_i$) and \eqref{eq:conj}, for every
$v\in\Lambda^\bullet V$
\[
  \Psi(\lambda L\,v)
   =\Psi\Bigl(\bigl(\lambda\textstyle\sum_i x_i\bigr)\cdot v\Bigr)
   =\psi\Bigl(\lambda\textstyle\sum_i x_i\Bigr)\cdot\Psi(v)
   =\bigl(e^{\lambda\omega}-1\bigr)\cdot\Psi(v).
\]
\end{proof}
We can now state the $K$-theory of $M_{|\lambda|}$ for all levels $\lambda$. 
\begin{theorem}\label{thm:cohomology-K}
Let $\lambda\in\ZZ\setminus\{0\}$ and $M_{|\lambda|}=\mu^{-1}(0)/\GG_\lambda$. 
\[
K^{0}(M_{|\lambda|})\ \cong\ \Z^{\binom{4g+1}{2g}}\oplus \bigoplus_{\substack{k\ \mathrm{even}\\ j\ge1}} \Z_{|\lambda|j}^{\,\varepsilon_{k-2}(j)},
\qquad
K^{1}(M_{|\lambda|})\ \cong\ \Z^{\binom{4g+1}{2g}}\oplus \bigoplus_{\substack{k\ \mathrm{odd}\\ j\ge1}} \Z_{|\lambda|j}^{\,\varepsilon_{k-2}(j)},
\]
where $\varepsilon_r(j)$ is the torsion multiplicity \eqref{eq:eps-closed}.
\end{theorem}

\begin{proof}
The $K$-theoretic Gysin sequence for the circle bundle $M_{|\lambda|}\to T^{4g}$ reads
\[
\cdots\to K^{i}(T^{4g})\xrightarrow{\,\cdot\,(1-e^{\lambda\omega})\,}K^{i}(T^{4g})
\xrightarrow{\pi^*}K^{i}(M_{|\lambda|})\xrightarrow{\pi_!}K^{i+1}(T^{4g})\to\cdots.
\]
Transporting the
computation to $\Lambda^\bullet V$,  the connecting map  becomes  wedge product by
$\Theta_\lambda=1-e^{\lambda\omega}$. By Lemma~\ref{lem:conj} there is a parity and
filtration preserving $\Z$-linear automorphism $\Psi$ of $\Lambda^\bullet V$ with
\begin{equation}\label{eq:theta}
\Theta_\lambda=-\,\Psi\circ(\lambda L)\circ\Psi^{-1},
\end{equation}
so on each parity $\ker\Theta_\lambda=\Psi(\ker \lambda L)\cong\ker(\lambda L)$ and $\im\Theta_\lambda=\Psi(\im \lambda L)$, and $\Psi$ induces
an isomorphism $\coker\Theta_\lambda\cong\coker(\lambda L)$ compatible with the degree
filtration. Since $\ker\Theta_\lambda\cong\ker(\lambda L)$ is free, being the kernel of a
homomorphism of free abelian groups, splitting the Gysin sequence by parity gives
\[
K^{0}(M_{|\lambda|})\cong\coker\big(\Theta_\lambda|_{\Lambda^{\mathrm{ev}}}\big)\oplus\ker\big(\Theta_\lambda|_{\Lambda^{\mathrm{odd}}}\big),
\quad
K^{1}(M_{|\lambda|})\cong\coker\big(\Theta_\lambda|_{\Lambda^{\mathrm{odd}}}\big)\oplus\ker\big(\Theta_\lambda|_{\Lambda^{\mathrm{ev}}}\big).
\]
Using \eqref{eq:theta} to replace $\Theta_\lambda$ by $\lambda L$, the degree $r$ summand of
$\coker(\lambda L)$ is $\coker(\lambda L_r)$, whose torsion is $\bigoplus_j\Z_{|\lambda|j}^{\varepsilon_r(j)}$
by Theorem~\ref{thm:cohomology}. Summing over $r$ of each
parity identifies $\Tors K^{i}$ with $\bigoplus_{k\equiv i}\Tors H^k(M_{|\lambda|})$. 

For the free part, by the Hard Lefschetz theorem over $\Q$, $L_r$ is injective for $r<2g$ and surjective for
$r\ge 2g$. Since $\Lambda^\bullet V$ is torsion-free, the ranks are the same over $\Z$, 
\begin{equation}\label{eq:rho}
\rank L_r=
\begin{cases}
\binom{4g}{r}, & 0\le r<2g,\\[2pt]
\binom{4g}{r+2}, & 2g\le r\le 4g-2.
\end{cases}
\end{equation}
These contributions telescope and using
$\binom{4g}{k}=\binom{4g}{4g-k}$ over a fixed parity, we get
\[
\rank K^0(M_{|\lambda|})=\rank K^1(M_{|\lambda|})=\binom{4g}{2g}+\binom{4g}{2g-1}=\binom{4g+1}{2g}.
\]
\end{proof}


\section{Interaction with T-duality} \label{RelT}

\subsection{Review of T-duality for principal circle bundles in background flux} \label{rev}

In \cite{BEM04a, BEM04b}, spacetime $Z$ was compactified in one direction.
More precisely, $Z$
is a principal $\bbT$-bundle over $X$

\[
\begin{CD}
\bbT @>>> Z \\
&& @V\pi VV \\
&& X \end{CD}
\]
 classified up to isomorphism by its first Chern class
{ $c_1(Z)\in H^2(X,\ZZ)$}. Assume that spacetime $Z$ is endowed with an $H$-flux which is
a representative in the
degree 3 Deligne cohomology of $Z$, that is
$H\in\Omega^3(Z)$ with integral periods (for simplicity, we drop factors of $\frac{1}{2\pi i}$),
together with
the following data. Consider a local trivialization $U_\alpha \times \TT$ of $Z\to X$, where
$\{U_\alpha\}$ is a good cover of $X$. Let $H_\alpha = H\Big|_{ U_\alpha \times \TT}
= d B_\alpha$, where $B_\alpha \in \Omega^2(U_\alpha \times \TT)$ and finally, $B_\alpha -B_\beta = F_{\alpha\beta}
\in \Omega^1(U_{\alpha\beta} \times \TT)$.
 Then the choice of $H$-flux entails that we are given a local trivialization
 as above and locally defined 2-forms $B_\alpha$ on it, together with closed 2-forms $F_{\alpha\beta}$ defined on 
 double overlaps,  that is, $(H, B_\alpha, F_{\alpha\beta})$. Also the first Chern class
 of $Z\to X$
  is represented in integral cohomology by $(F, A_\alpha)$ where
$\{A_\alpha\}$ is a connection 1-form on $Z\to X$ and $F = dA_\alpha$ is the curvature 2-form of $\{A_\alpha\}$.

The {{T-dual}}  is another principal
$\bbT$-bundle over $M$, denoted by $\widehat Z$,
  {}
\[
\begin{CD}
\widehat \bbT @>>> \widehat Z \\
&& @V\widehat \pi VV     \\
&& X \end{CD}
\]
To define it, we see that $\pi_* (H_\alpha) = d \pi_*(B_\alpha) = d {\widehat A}_\alpha$,
 so that $\{{\widehat A}_\alpha\}$ is a connection 1-form whose curvature $ d {\widehat A}_\alpha = \widehat F_\alpha =  \pi_*(H_\alpha)$
 that is, $\widehat F = \pi_* H$. So let $\widehat Z$ denote the principal
$\bbT$-bundle over $M$ whose first Chern class is  $\,\, c_1(\widehat Z) = [\pi_* H, \pi_*(B_\alpha)] \in H^2(X; \ZZ) $.

The Gysin
sequence for $Z$ enables us to define a T-dual $H$-flux
$[\widehat H]\in H^3(\widehat Z,\ZZ)$, satisfying
\[
c_1(Z) = \widehat \pi_* \widehat H \,,
\]
 where $\pi_* $
and similarly $\widehat\pi_*$, denote the pushforward maps.
Note that $ \widehat H$ is not fixed by this data, since any integer
degree 3 cohomology class on $X$ that is pulled back to $\widehat Z$
also satisfies the requirements. However, $ \widehat H$ is
determined uniquely (up to cohomology) upon imposing
the condition $[H]=[\widehat H]$ on the correspondence space $Z\times_X \widehat Z$
as will be explained now.

The {\em correspondence space} (sometimes called the doubled space) is defined as
$$
Z\times_X  \widehat Z = \{(x, \widehat x) \in Z \times \widehat Z: \pi(x)=\widehat\pi(\widehat x)\}.
$$
Then we have the following commutative diagram,
\begin{equation*} \label{eqn:correspondence}
\xymatrix @=6pc @ur { (Z, [H]) \ar[d]_{\pi} &
(Z\times_X  \widehat Z, [H]=[\widehat H]) \ar[d]_{\widehat p} \ar[l]^{p} \\ X & (\widehat Z, [\widehat H])\ar[l]^{\widehat \pi}}
\end{equation*}
By requiring that
$$
p^*[H]={\widehat p}^*[\widehat H] \in H^3(Z\times_X  \widehat Z, \ZZ),
$$
determines $[\widehat H] \in H^3(  \widehat Z, \ZZ)$  uniquely, via an application of the Gysin sequence.
An alternate way to see this is is explained below.

Let $(H, B_\alpha, F_{\alpha\beta}, L_{\alpha\beta})$ denote a gerbe with connection on $Z$.
We also choose a connection 1-form $A$ on $Z$.
Then define 
\be \label{newconn} \widehat A_\alpha = -\iota_v B_\alpha\ee 
on the chart $U_\alpha$ and
the connection 1-form $\widehat A= \widehat A_\alpha +d\widehat\theta_\alpha$
on the chart $U_\alpha\times  \widehat \TT$, where $\theta_\alpha$ is the coordinate on $\widehat \TT$. 
In this way we get a T-dual circle bundle
$\widehat Z \to X$ with connection 1-form $\widehat A$.

Without loss of generality, we can assume that $H$ is $\TT$-invariant. Consider
\[
\Omega = H - A\wedge F_{\widehat A}
\]
where  $F_{\widehat A} = d {\widehat A}$ and $F_{A} = d {A}$ are the curvatures of $A$
and $\widehat A$ respectively. One checks that the contraction $i_v(\Omega)=0$ and
the Lie derivative $L_v(\Omega)=0$ so that $\Omega$ is a basic 3-form on $Z$, that is
$\Omega$ comes from the base $X$.

Setting
\be \label{split}
\widehat H = F_A\wedge {\widehat A} + \Omega
\ee
this defines the T-dual flux 3-form. One verifies that $\widehat H$ is a closed 3-form on $\widehat Z$.
It follows that on the correspondence space, one has as desired,
\[
\widehat H = H + d (A\wedge \widehat A ).
\]

Our next goal is to determine the T-dual curving or B-field.
The Buscher rules imply that on the open sets $U_\alpha \times \TT\times \widehat \TT$ of the
correspondence space $Z\times_X \widehat Z$, one has
\begin{equation} \label{buscher}
\widehat B_\alpha = B_\alpha + A\wedge \widehat A - d\theta_\alpha \wedge d\widehat \theta _\alpha\,,
\end{equation}
Note that
\[
\iota_v \widehat B_\alpha = \iota_v
\left( B_\alpha + A\wedge \widehat A - d\theta_\alpha \wedge d\widehat \theta _\alpha\right) =
-\widehat A_\alpha + \widehat A - d\widehat \theta_\alpha = 0
\]
so that $\widehat B_\alpha$ is indeed a 2-form on $\widehat Z$ and not just on the correspondence
space. Obviously, $d \widehat B_\alpha = \widehat H$. Following the descent equations one arrives at the complete
T-dual gerbe with connection, $(\widehat H, \widehat B_\alpha, \widehat F_{\alpha\beta}, \widehat L_{\alpha\beta})$.
cf. \cite{BMPR}.

The rules
for transforming the Ramond-Ramond (RR) fields can be encoded in the {\cite{BEM04a, BEM04b}} generalization of
{\em Hori's formula}
  {}
\begin{equation} \label{eqn:Hori}
T_*G =  \int_{\bbT} e^{ -A \wedge \widehat A }\ G \,,
\end{equation}
 where $G \in \Omega^\bullet(Z)^\bbT$ is the total RR fieldstrength,
\begin{center}
$G\in\Omega^{even}(Z)^\bbT \quad$ for {   { Type IIA}};\\
$G\in\Omega^{odd}(Z)^\bbT \quad$ for {   { Type IIB}},\\
\end{center}
and where the right hand side of equation \eqref{eqn:Hori} is an invariant differential form on $Z\times_X\widehat Z$, and
the integration is along the $\bbT$-fiber of $Z$.

Recall that the twisted cohomology
is defined as the cohomology of the complex
$$H^\bullet(Z, H) = H^\bullet(\Omega^\bullet(Z), d_H=d+ H\wedge).$$
By the identity \eqref{eqn:Hori}, $T_*$ maps $d_H$-closed forms $G$ to $d_{\widehat
H}$-closed forms $T_*G$.
 So T-duality $T_*$  induces a map on twisted cohomologies,
$$
T : H^\bullet(Z, H) \to H^{\bullet +1}(\widehat Z, \widehat H).
$$
Define the Riemannian metrics on $Z$ and $\widehat Z$ respectively by
\be \label{R}
h_{A, R}=\pi^*g_X+R^2\, A\odot A,\qquad \widehat h_{\widehat A, 1/R}=\widehat\pi^*g_X+1/{R^2} \,\widehat A\odot\widehat A.
\ee
where $h_X$ is a Riemannian metric on $X$.
 Then $h_{A, R}$ is $\TT$-invariant and the length of each circle fibre is $R$; $\widehat h_{\widehat A, 1/R}$
 is $\widehat\TT$-invariant and the length of each circle fibre is $1/R$.

 The following theorem summarizes the main consequence of T-duality for principal circle bundles in a background flux.

\begin{theorem}[T-duality isomorphism \cite{BEM04a,BEM04b}]\label{thm:T-duality}
In the notation above, and
with the above choices of Riemannian metrics and flux forms, the map \eqref{eqn:Hori}
$$
T\colon\Omega^{\overline k}(Z)^\TT\to\Omega^{\overline{k+1}}(\widehat Z)^{\widehat\TT},
$$
for $k=0,1$, (where $\overline k$ denotes the parity of $k$) are isometries, inducing isomorphisms on twisted
cohomology groups,
\[
T : H^\bullet(Z, H) \stackrel{\cong}{\longrightarrow} H^{\bullet +1}(\widehat Z, \widehat H).
\]
Therefore under T-duality one has the exchange,
\begin{center}
{$R \Longleftrightarrow 1/R$}\quad and \quad
{{   {background H-flux} $\Longleftrightarrow$   {Chern class}}}
\end{center}
Moreover there is also an isomorphism of twisted K-theories,
\[
T : K^\bullet(Z, H) \to K^{\bullet +1}(\widehat Z, \widehat H),
\]
such that the following diagram commutes,
\[
\xymatrix @=4pc 
{ K^\bullet(Z, H) \ar[d]_{Ch_H}  \ar[r]^{T} & K^{\bullet +1}(\widehat Z, \widehat H)  \ar[d]^{Ch_{\widehat H}}   
\\ H^\bullet(Z, H)  \ar[r]_{T}& H^{\bullet +1}(\widehat Z, \widehat H)}
\]
\end{theorem}

\subsection{T-duality of the  twisted Yang-Mills functional } \label{TYM}

Recall that the metrics on the T-dual spaces are rescaled as in \eqref{R}:
\[
h_{A,R}
=
\pi^*h_X+R^2A\odot A,
\qquad
\widehat h_{\widehat A,1/R}
=
\widehat\pi^*h_X+\frac1{R^2}\widehat A\odot\widehat A.
\]

To incorporate the radius parameter \(R\), we introduce the rescaled twisted Yang--Mills functional
\[
\mathrm{YM}(A,B)_R
:=
\sqrt R
\left(
\int_\Sigma |F_A|^2\,d\operatorname{vol}_h
+
\int_Z H\wedge *_{h_{A,R}}H
\right),
\]
and define
\(
\mathrm{YM}(\widehat A,\widehat B)_{1/R}
\)
analogously on the dual side.

The normalization by \(\sqrt R\) is chosen so that the functional transforms symmetrically under the exchange
\begin{center}
$R \Longleftrightarrow 1/R$.
\end{center}

The following result shows that this normalization is precisely the one compatible with T-duality.

\begin{theorem}\label{mainYM}
One has
\[
\mathrm{YM}(A,B)_R
=
\mathrm{YM}(\widehat A,\widehat B)_{1/R}.
\]
\end{theorem}

\begin{proof}
Since \(\Sigma\) is two-dimensional, we may write
\[
F_A=f_A\,d\operatorname{vol}_h.
\]

Using the decomposition \eqref{split}, together with the fact that
\(\widehat H\) is basic, we obtain
\[
\widehat H
=
\widehat A\wedge F_A
=
\widehat A\wedge f_A\,d\operatorname{vol}_h.
\]

By the definition of the rescaled metric
\(
\widehat h_{\widehat A,1/R},
\)
the corresponding volume form satisfies
\[
\widehat A\wedge d\operatorname{vol}_h
=
R\,d\operatorname{vol}_{\widehat h_{\widehat A,1/R}},
\]
and therefore
\[
\widehat H
=
R\,f_A\,d\operatorname{vol}_{\widehat h_{\widehat A,1/R}}.
\]

It follows that
\[
\int_{\widehat Z}
\widehat H\wedge
*_{\widehat h_{\widehat A,1/R}}
\widehat H
=
R
\int_\Sigma
|F_A|^2\,d\operatorname{vol}_{h_{\widehat A,1/R}}.
\]

Interchanging the roles of the dual spaces yields
\[
\int_Z
H\wedge
*_{h_{A,R}}
H
=
\frac{1}{R}
\int_\Sigma
|F_{\widehat A}|^2\,d\operatorname{vol}_h.
\]

Substituting these identities into the definition of
\(
\mathrm{YM}(\widehat A,\widehat B)_{1/R}
\),
we obtain
\begin{align*}
\mathrm{YM}(\widehat A,\widehat B)_{1/R}
&=
\frac1{\sqrt R}
\left(
\int_\Sigma
|F_{\widehat A}|^2\,d\operatorname{vol}_h
+
\int_{\widehat Z}
\widehat H
\wedge
*_{\widehat h_{\widehat A,1/R}}
\widehat H
\right)
\\
&=
\frac1{\sqrt R}
\left(
R\int_Z
H\wedge *_{h_{A,R}}H
+
R\int_\Sigma
|F_A|^2\,d\operatorname{vol}_h
\right)
\\
&=
\sqrt R
\left(
\int_\Sigma
|F_A|^2\,d\operatorname{vol}_h
+
\int_Z
H\wedge *_{h_{A,R}}H
\right)
\\
&=
\mathrm{YM}(A,B)_R.
\end{align*}
\end{proof}

\subsection{T-dualizable moduli stack of flat mixed fields} \label{Tmoduli}

In this section we examine the action of T-duality on the space of flat mixed fields. 
As recalled in Section~\ref{rev}, the T-duality correspondence carries a triple \((Z,A,B)\) to its dual \((\widehat Z,\widehat A,\widehat B)\).
We first show that it defines a self-map on $\mu^{-1}(0)$.

\begin{lemma}\label{lem:buscher-level}
The Buscher rules \eqref{newconn}--\eqref{buscher}  induce a self-map
\[
  \mathcal T:\mu^{-1}(0)\longrightarrow \mu^{-1}(0),\qquad
  \mathcal T(a_0,b_0,b_1)=(-b_0,\,-a_0,\,b_1+a_0\wedge b_0),
\]
that is independent of the level $\lambda$ and satisfies $\mathcal T^2=\mathrm{id}$.
\end{lemma}

\begin{proof}
 On the flat locus, we trivialize $Z$ and write $A=a_0+d\theta$, $B=b_1+A\wedge b_0$
as in \eqref{expression}, so that $\iota_vB=b_0$ by \eqref{eq:iota-B}. Since
$H=0$ and $F_A=0$, by
\eqref{newconn} the dual connection is $\widehat A=-b_0+d\widehat\theta$, hence
$\widehat a_0=-b_0$, and $\Omega=H-A\wedge F_{\widehat A}=0$, so
$\widehat H=0$ by \eqref{split}. The dual curving \eqref{buscher} is
\[
  \widehat B
  = B+A\wedge\widehat A-d\theta\wedge d\widehat\theta
  = b_1+A\wedge b_0+A\wedge(-b_0+d\widehat\theta)-d\theta\wedge d\widehat\theta
  = b_1+a_0\wedge d\widehat\theta,
\]
using $A=a_0+d\theta$. Contracting with the dual vertical field
$\widehat v=\partial_{\widehat\theta}$ gives $\widehat b_0=\iota_{\widehat v}\widehat B=-a_0$,
and therefore
\[
  \widehat b_1
  = \widehat B-\widehat A\wedge\widehat b_0
  = \bigl(b_1+a_0\wedge d\widehat\theta\bigr)-(-b_0+d\widehat\theta)\wedge(-a_0)
  = b_1+a_0\wedge b_0,
\]
the terms $a_0\wedge d\widehat\theta$ cancelling. Hence
$\mathcal T(a_0,b_0,b_1)=(-b_0,-a_0,b_1+a_0\wedge b_0)$.

By Proposition~\ref{YMprop}(i), together with \eqref{eq:iota-twist},
\eqref{eq:d-twist} and Lemma~\ref{lem:cocycle-repair}, both $F_A$ and $H$ are
strictly invariant under the level-$\lambda$ action, so the flat locus
$\mu^{-1}(0)$ is the same subset of $\Ca$ for every $\lambda$. Under a level-$\lambda$
gauge transformation the dual connection changes by
\[
  \widehat A=-\iota_vB
  \ \longmapsto\
  -\iota_v\bigl(B+\eta+\lambda\,u\wedge\iota_vB\bigr)
  =\widehat A-\iota_v\eta ,
\]
again by \eqref{eq:iota-twist}. Since $\iota_v\eta$ is closed with integral
periods \eqref{eq:iota-eta}, this is an ordinary gauge transformation of the dual
bundle,  $F_{\widehat A}$ is unchanged. Thus $\mathcal T$ maps $\mu^{-1}(0)$ to itself at every level.

Applying $\mathcal T$ twice,
\[
  \mathcal T^2(a_0,b_0,b_1)
  =\mathcal T(-b_0,-a_0,b_1+a_0\wedge b_0)
  =\bigl(a_0,\ b_0,\ b_1+a_0\wedge b_0+(-b_0)\wedge(-a_0)\bigr)
  =(a_0,b_0,b_1),
\]
since $(-b_0)\wedge(-a_0)=b_0\wedge a_0=-a_0\wedge b_0$. Hence $\mathcal T^2=\mathrm{id}$.
\end{proof}

What the level $\lambda$ does change is the \emph{equivariance} of $\mathcal T$, which is precisely the
content of the following definition. We allow, a priori, the two sides of the correspondence to
be twisted at different levels $\lambda$ on $Z$ and $\lambda'$ on $\widehat Z$; the analysis
below shows that consistency forces $\lambda'=\lambda$.

\begin{definition}\label{def:Tdualisable}
We say that a gauge transformation \((g,(L,\nabla^{L}))\in\GG_\lambda\) is
\emph{T-dualizable at \((A,B)\)} if there exists a gauge transformation
\((\widehat g,(\widehat L,\nabla^{\widehat L}))\in\GG_{\lambda'}\) on the dual side such that
\begin{equation}\label{eq:T-equivariance}
  \mathcal T \bigl((g,(L,\nabla^{L}))\cdot(A,B)\bigr)
   =
  (\widehat g,(\widehat L,\nabla^{\widehat L}))\cdot \mathcal T(A,B).
\end{equation}
\end{definition}

Denote by \(\GG^T_{(A,B)}< \GG_\lambda\) the subgroup consisting of all T-dualizable gauge transformations at \((A,B)\), i.e.
\[
  \GG^T_{(A,B)}
  :=
  \Bigl\{(g,(L,\nabla^{L}))\in \GG_\lambda \,\Bigm|\, (g,(L,\nabla^{L})) \text{ is T-dualizable at } (A,B)\Bigr\}.
\]

Introduce the equivalence relation $\underset{T}{\sim}$ on $\mu^{-1}(0)$ by declaring
\[
(A,B)\underset{T}{\sim}(A',B')
\quad\Longleftrightarrow\quad
(A',B')
=
h\cdot(A,B)
\text{ for some }
h\in\GG^{T}_{(A,B)}.
\]

Recall from
Lemma~\ref{lem:finite-model} the level-$\lambda$ cocycle
$\phi^\lambda_{(\mathbf m,\mathbf n)}(\mathbf x,\mathbf y)
=\mathbf n\Theta_g\mathbf x^T+(\lambda-1)\mathbf m\Theta_g\mathbf y^T$, and  
\begin{equation}\label{eq:phi-zero}
  \phi^{0}_{(\mathbf m,\mathbf n)}(\mathbf x,\mathbf y)
  =\mathbf n\Theta_g\mathbf x^T-\mathbf m\Theta_g\mathbf y^T ,
\end{equation}
is the cocycle of the untwisted action.

\begin{proposition}\label{prop:obstruction}
Let \((A,B)\in\mu^{-1}(0)\) have classes $[a_0]=\mathbf x$, $[b_0]=\mathbf y$, and let
$h=(g,(L,\nabla^L))\in\GG_\lambda$ have $[u]=\mathbf m$, $[\iota_v\eta]=\mathbf n$. Then
\eqref{eq:T-equivariance} forces
\[
  [\widehat u]=-\mathbf n,\qquad [\iota_{\widehat v}\widehat\eta]=-\mathbf m ,
\]
independently of $\lambda$ and $\lambda'$, and the residual obstruction to
\eqref{eq:T-equivariance} in $\RR/\ZZ$ is
\begin{equation}\label{eq:residual}
  (\lambda-1)\,\mathbf m\Theta_g\mathbf y^T-(\lambda'-1)\,\mathbf n\Theta_g\mathbf x^T .
\end{equation}
Hence,
\begin{enumerate}[(i)]
\item the obstruction is proportional to $\phi^{0}$, uniformly in $(\mathbf m,\mathbf n)$, if and
only if $\lambda'=\lambda$, in which case \eqref{eq:residual} equals
$-(\lambda-1)\,\phi^{0}_{(\mathbf m,\mathbf n)}(\mathbf x,\mathbf y)$ modulo $\ZZ$,
\item the obstruction vanishes identically  if and only if $\lambda=\lambda'=1$.
\end{enumerate}
\end{proposition}

\begin{proof}
Write both sides of \eqref{eq:T-equivariance} in the coordinates of
Lemma~\ref{lem:finite-model}, using Lemma~\ref{lem:buscher-level} for $\mathcal T$.
The left-hand side is $\mathcal T$ applied to
$(\mathbf x+\mathbf m,\ \mathbf y+\mathbf n,\
[t+\phi^{\lambda}_{(\mathbf m,\mathbf n)}(\mathbf x,\mathbf y)])$, namely
\[
  \Bigl(-(\mathbf y+\mathbf n),\ -(\mathbf x+\mathbf m),\
  \bigl[t+\phi^{\lambda}_{(\mathbf m,\mathbf n)}(\mathbf x,\mathbf y)
  +(\mathbf x+\mathbf m)\Theta_g(\mathbf y+\mathbf n)^T\bigr]\Bigr).
\]
The right-hand side is $\widehat h$ applied to $\mathcal T(\mathbf x,\mathbf y,[t])
=(-\mathbf y,-\mathbf x,[t+\mathbf x\Theta_g\mathbf y^T])$, namely
\[
  \Bigl(-\mathbf y+\widehat{\mathbf m},\ -\mathbf x+\widehat{\mathbf n},\
  \bigl[t+\mathbf x\Theta_g\mathbf y^T
  +\phi^{\lambda'}_{(\widehat{\mathbf m},\widehat{\mathbf n})}(-\mathbf y,-\mathbf x)\bigr]\Bigr).
\]
Matching the first two components gives $\widehat{\mathbf m}=-\mathbf n$ and
$\widehat{\mathbf n}=-\mathbf m$, as claimed.
Substituting and subtracting, the fibre components differ by
\[
  \phi^{\lambda}_{(\mathbf m,\mathbf n)}(\mathbf x,\mathbf y)
  +(\mathbf x+\mathbf m)\Theta_g(\mathbf y+\mathbf n)^T-\mathbf x\Theta_g\mathbf y^T
  -\phi^{\lambda'}_{(-\mathbf n,-\mathbf m)}(-\mathbf y,-\mathbf x).
\]
Expanding both cocycles and using $\mathbf a\Theta_g\mathbf b^T=-\mathbf b\Theta_g\mathbf a^T$
throughout, the terms $\mathbf n\Theta_g\mathbf x^T$ and $\mathbf x\Theta_g\mathbf n^T$ cancel
in pairs and  $\mathbf m\Theta_g\mathbf n^T$ is an integer since $\Theta_g$ is integral, so one is left with \eqref{eq:residual}. For (i), the
expression \eqref{eq:residual} is a multiple of
$\phi^{0}=\mathbf n\Theta_g\mathbf x^T-\mathbf m\Theta_g\mathbf y^T$ for all
$(\mathbf m,\mathbf n,\mathbf x,\mathbf y)$ precisely when the coefficients of
$\mathbf m\Theta_g\mathbf y^T$ and  $\mathbf n\Theta_g\mathbf x^T$ are opposite, i.e.\
$\lambda-1=\lambda'-1$, and then equal to $-(\lambda-1)\phi^{0}$ modulo $\ZZ$. For (ii), the
$(\mathbf x,\mathbf y)$-dependent part of \eqref{eq:residual} vanishes identically iff
$\lambda=1$ and $\lambda'=1$.
\end{proof}

\begin{remark}\label{rem:matched-levels}
Since $\mathcal T$ exchanges $A$ with $\widehat A$, and hence the roles of
$\GG_1$ and $\GG_2$, requiring $\mathcal T$  to relate the level-$\lambda$-theory to a dual theory of the same form forces the two sides to be twisted at the same level.
For $\lambda'\ne\lambda$ the $\mathbf x$- and $\mathbf y$-obstructions in
\eqref{eq:residual} decouple and one obtains a coherent but finer stratification with no
description of the quotient in terms of the single homomorphism \eqref{eq:l-map} below.
\end{remark}

Assuming $\lambda'=\lambda$ henceforth, Proposition~\ref{prop:obstruction} shows that the T-dualizability
condition at $(A,B)$ reads
\[
  (\lambda-1)\,\phi^{0}_{(\mathbf m,\mathbf n)}(\mathbf x,\mathbf y)\in\ZZ .
\]
Accordingly we introduce, for $(\mathbf x,\mathbf y)\in\RR^{2g}\times\RR^{2g}$,
\[
G_{(\mathbf{x},\mathbf{y})}
:=
\left\{
(\mathbf m,\mathbf n)
\in
\mathbb Z^{2g}\times\mathbb Z^{2g}
\;\Big|\;
  (\lambda-1)\,\phi^{0}_{(\mathbf m,\mathbf n)}(\mathbf x,\mathbf y) 
\in\mathbb Z
\right\},
\]
and note that $G_{(\mathbf x,\mathbf y)}=G_{(\mathbf x+\mathbf a,\mathbf y+\mathbf b)}$ for
$(\mathbf a,\mathbf b)\in\ZZ^{2g}\times\ZZ^{2g}$.

\medskip

Let \(\GG_0< \GG_\lambda\) be the subgroup
\[
  \GG_0
  :=
  \Bigl\{(g,(L,\nabla^{L}))\in \GG_\lambda \,\Bigm|\,
    g=e^{2\pi i f_1},\ 
    \mathrm{hol}^{\,v}(\nabla^{L})=e^{2\pi i f_2},
    \ \text{for some } f_1,f_2\in C^\infty(\Sigma,\RR)
  \Bigr\},
\]
which by the proof of Lemma~\ref{lem:finite-model} is independent of $\lambda$.
There is an exact sequence
\[
  1 \longrightarrow \GG_0 \xrightarrow{\;i\;} \GG_\lambda
    \xrightarrow{\;h\;} \ZZ^{2g}\times \ZZ^{2g} \longrightarrow 0,
\]
where
\[
  h:\GG_\lambda\longrightarrow \ZZ^{2g}\times \ZZ^{2g},\qquad
  (g,(L,\nabla^{L})) \longmapsto (\mathbf m,\mathbf n)
  =(m_1,\dots,m_{2g},\,n_1,\dots,n_{2g}),
\]
is an epimorphism, with the integers \(m_i,n_i\) defined as in~\eqref{expression1} after fixing a symplectic basis \(\{\theta_1,\dots,\theta_g,\eta_1,\dots,\eta_g\}\) of \(H^{1}(\Sigma;\RR)\). Note
that $h$ is a homomorphism onto an abelian group even though $\GG_\lambda$ is a semi-direct
product, since by Lemma~\ref{lem:semidirect} the twist alters $\eta$ by
$\lambda\,u\wedge\iota_v\eta$, whose class $\lambda\,\mathbf m\Theta_g\mathbf n^T$ is integral
and therefore invisible to $[\iota_v\eta]$.

\medskip

Given \((A,B)\in \mu^{-1}(0)\), express \(A\) and \(B\) as in~\eqref{expression},
Write the cohomology classes as
\[ [a_0]=\sum_{i=1}^{g}\bigl(x_i\theta_i + x_{i+g}\eta_i\bigr)=(x_1, \cdots, x_{2g}),\qquad
[b_0]=\sum_{i=1}^{g}\bigl(y_i\theta_i + y_{i+g}\eta_i\bigr)=(y_1, \cdots, y_{2g}).
\]

We then introduce the homomorphism
\begin{equation}\label{eq:l-map}
  l_{(A,B)}:\ \ZZ^{2g}\times \ZZ^{2g}\longrightarrow \RR/\ZZ,\qquad
  (\mathbf m,\mathbf n)\longmapsto \bigl\{(\lambda-1)\bigl(\mathbf n\,\Theta_g[a_0]^{T}
  - \mathbf m\,\Theta_g[b_0]^{T}\bigr)\bigr\},
\end{equation}
where the notation \(\{\cdot\}\) denotes the class modulo \(\ZZ\), and \(\Theta_g\), \(a_0\), \(b_0\) are as in Section~\ref{rev}. Thus $l_{(A,B)}=(\lambda-1)\cdot l^{0}_{(A,B)}$, where
$l^{0}_{(A,B)}$ is the homomorphism of the untwisted theory  and $G_{(\mathbf x,\mathbf y)}=\ker l_{(A,B)}$.

For $(\mathbf x,\mathbf y)\in\RR^{2g}\times\RR^{2g}$, define the relation $\approx_T$ on $\RR^{2g}\times\RR^{2g}\times\RR/\ZZ$ by
\[
  (\mathbf x,\mathbf y,[t])\ \approx_T\ (\mathbf x',\mathbf y',[t'])
  \iff
  \exists\,(\mathbf m,\mathbf n)\in G_{(\mathbf x,\mathbf y)}:\
    \mathbf x'=\mathbf x+\mathbf m,\ \ \mathbf y'=\mathbf y+\mathbf n,\ \
    [t']=\bigl[t+\phi^{\lambda}_{(\mathbf m,\mathbf n)}(\mathbf x,\mathbf y)\bigr].
\]

\begin{theorem}\label{main2}
(i) For every \( (A,B)\in\mu^{-1}(0) \),
\[
\GG^{T}_{(A,B)} \;\cong\; \GG_{0} \oplus \ker l_{(A,B)}.
\]

(ii) The relation $\approx_T$ is an equivalence relation, and the map
$\Phi(A,B)=\bigl([a_0],[b_0],[\textstyle\int_\Sigma b_1]\bigr)$ descends to a bijection
\[
  \mu^{-1}(0)\big/\underset{T}{\sim}
  \ \cong\
  \bigl(\RR^{2g}\times\RR^{2g}\times\RR/\ZZ\bigr)\big/\approx_T .
\]

(iii) The T-duality map $\mathcal T:\mu^{-1}(0)\to\mu^{-1}(0)$ induces a well-defined
involution $\mathcal T_*$ on the quotient stack, which under the identification in (ii) is
\[
  \mathcal T_*\bigl[(\mathbf x,\mathbf y),[t]\bigr]
  =\bigl[(-\mathbf y,-\mathbf x),\ [t+\mathbf x\Theta_g\mathbf y^T]\bigr].
\]
\end{theorem}

\begin{proof}
\emph{(i)} As in the proof of Lemma~\ref{lem:finite-model}, write a flat mixed field as
$A=a_0+d\theta$, $B=b_1+A\wedge b_0$ and identify $(A,B)$ with the triple $(a_0,b_0,b_1)$.
By Lemma~\ref{lem:buscher-level}, $\mathcal T(a_0,b_0,b_1)=(-b_0,-a_0,b_1+a_0\wedge b_0)$,
and by \eqref{afterchange} the level-$\lambda$ action of $h=(g,(L,\nabla^L))$ is
\[
  h\cdot(a_0,b_0,b_1)=\Bigl(a_0+u,\ b_0+\iota_v\eta,\
  b_1+\bigl(\eta-A\wedge\iota_v\eta-u\wedge\iota_v\eta\bigr)+(\lambda-1)\,u\wedge b_0\Bigr).
\]
Imposing \eqref{eq:T-equivariance} and comparing the first two components gives, as in
Proposition~\ref{prop:obstruction},
$\widehat g^{\,-1}d\widehat g=-\iota_v\eta$ and $\iota_{\widehat v}\widehat\eta=-u$,
the exchange of the two gauge parameters; this imposes the integrality conditions on the
corresponding closed $1$-forms and is independent of the level. The remaining equation is an
equality of $2$-forms on $\Sigma$ modulo integral classes, whose obstruction is, by
Proposition~\ref{prop:obstruction},
\[
  (\lambda-1)\bigl(\mathbf n\Theta_g[a_0]^T-\mathbf m\Theta_g[b_0]^T\bigr)\in\ZZ ,
\]
that is $(\mathbf m,\mathbf n)\in\ker l_{(A,B)}$ in the notation of \eqref{eq:l-map}. The
subgroup of gauge transformations with trivial integral part is $\GG_0$, whence
$\GG^{T}_{(A,B)}\cong\GG_0\oplus\ker l_{(A,B)}$.

\emph{(ii)} We first check $\approx_T$ is an equivalence relation. Reflexivity holds since
$(\mathbf 0,\mathbf 0)\in G_{(\mathbf x,\mathbf y)}$ and $\phi^{\lambda}_{(\mathbf 0,\mathbf 0)}=0$;
the admissible set is leaf-invariant, $G_{(\mathbf x,\mathbf y)}=G_{(\mathbf x+\mathbf a,\mathbf y+\mathbf b)}$
for $(\mathbf a,\mathbf b)\in\ZZ^{2g}\times\ZZ^{2g}$, since the defining quantity changes by
$(\lambda-1)(\mathbf n\Theta_g\mathbf a^T-\mathbf m\Theta_g\mathbf b^T)\in\ZZ$; and composing
$(\mathbf m,\mathbf n)$ then $(\mathbf m',\mathbf n')$ shifts the fibre by
$\phi^{\lambda}_{(\mathbf m+\mathbf m',\mathbf n+\mathbf n')}(\mathbf x,\mathbf y)
+(\lambda-1)\mathbf m'\Theta_g\mathbf n^T-\mathbf m\Theta_g\mathbf n'^{T}$, whose last two terms
lie in $\ZZ$, giving transitivity and (with $(\mathbf m',\mathbf n')=(-\mathbf m,-\mathbf n)$)
symmetry.

Now represent every element of $\mu^{-1}(0)$ by a triple $(a_0,b_0,b_1)$ with $a_0,b_0$ closed.
Fixing a symplectic basis and setting $[a_0]=\mathbf x$, $[b_0]=\mathbf y$,
$[t]=[\int_\Sigma b_1]\in\RR/\ZZ$ defines
\[
  \Phi:\mu^{-1}(0)\longrightarrow\RR^{2g}\times\RR^{2g}\times\RR/\ZZ,\qquad
  (A,B)\longmapsto(\mathbf x,\mathbf y,[t]).
\]
Two triples have the same image under $\Phi$ iff they differ by an element of $\GG_0$. Indeed,
$\GG_0$ consists exactly of the transformations with $[u]=[\iota_v\eta]=0$, which by
\eqref{afterchange} fix $(\mathbf x,\mathbf y)$ and shift $[\int_\Sigma b_1]$ by an integer,
hence fix $[t]$. Conversely a transformation fixing $(\mathbf x,\mathbf y,[t])$ has
$[u]=[\iota_v\eta]=0$ and so lies in $\GG_0$. As $b_1\in\Omega^2(\Sigma)$ is unconstrained on
$\mu^{-1}(0)$, every value of $[t]$ is attained, so $\Phi$ is a bijection
\[
  \mu^{-1}(0)/\GG_0\ \xrightarrow{\ \cong\ }\ \RR^{2g}\times\RR^{2g}\times\RR/\ZZ .
\]
It remains to pass from $\GG_0$ to the full T-dualizable group. By part~(i) the residual group
$\GG^{T}_{(A,B)}/\GG_0\cong\ker l_{(A,B)}=G_{(\mathbf x,\mathbf y)}$ acts on the image through
its integral part $(\mathbf m,\mathbf n)$, sending, by \eqref{afterchange}, $(\mathbf x,\mathbf y)$
to $(\mathbf x+\mathbf m,\mathbf y+\mathbf n)$ and shifting $[t]$ by
$\phi^{\lambda}_{(\mathbf m,\mathbf n)}(\mathbf x,\mathbf y)$; this is precisely $\approx_T$.
Since $\underset{T}{\sim}$ is generated by $\GG_0$ together with these integral-part
transformations, $\Phi$ carries it exactly to $\approx_T$, and therefore descends to a
bijection
\[
  \mu^{-1}(0)\big/\underset{T}{\sim}\ \cong\
  \bigl(\RR^{2g}\times\RR^{2g}\times\RR/\ZZ\bigr)\big/\approx_T .
\]

\emph{(iii)} Suppose $(A,B)\underset{T}{\sim}(A',B')$, say $(A',B')=h\cdot(A,B)$ with
$h\in\GG^T_{(A,B)}$, and let $\widehat h$ realise \eqref{eq:T-equivariance}. Then
$\mathcal T(A',B')=\widehat h\cdot\mathcal T(A,B)$, and it remains to check
$\widehat h\in\GG^T_{\mathcal T(A,B)}$. By Proposition~\ref{prop:obstruction} the obstruction
for $\widehat h$ at $\mathcal T(A,B)$, whose classes are $[\widehat a_0]=-\mathbf y$,
$[\widehat b_0]=-\mathbf x$ and whose integral part is
$(\widehat{\mathbf m},\widehat{\mathbf n})=(-\mathbf n,-\mathbf m)$, is
\[
  (\lambda-1)\bigl(\widehat{\mathbf n}\Theta_g[\widehat a_0]^T
  -\widehat{\mathbf m}\Theta_g[\widehat b_0]^T\bigr)
  =-(\lambda-1)\,\phi^{0}_{(\mathbf m,\mathbf n)}(\mathbf x,\mathbf y),
\]
which differs by a sign from the integral quantity of part~(i) and so also lies in $\ZZ$.
Hence $\widehat h\in\GG^T_{\mathcal T(A,B)}$, so $\mathcal T(A',B')\underset{T}{\sim}\mathcal T(A,B)$
and $\mathcal T$ induces a well-defined map $\mathcal T_*$ on the quotient. In the coordinates
of part~(ii), $\mathcal T$ sends $(\mathbf x,\mathbf y,[t])$ to
$(-\mathbf y,-\mathbf x,[t+\mathbf x\Theta_g\mathbf y^T])$ by Lemma~\ref{lem:buscher-level},
which is the stated formula, and it respects $\approx_T$ by the computation just given.
Finally $\mathcal T_*^2=\mathrm{Id}$ follows from
$\mathbf x\Theta_g\mathbf y^T+(-\mathbf y)\Theta_g(-\mathbf x)^T=0$ by skew-symmetry.
\end{proof}

The following theorem characterises how the structure of the T-dualizable subgroup, and hence of the quotient stack, varies with the
level. 

\begin{theorem}\label{thm:level-stack}
Let $\lambda\in\ZZ\setminus\{0\}$.
\begin{enumerate}[(i)]
\item If $\lambda=1$,  the relation $\underset{T}{\sim}$ becomes ordinary
gauge equivalence, so 
\[
  \mu^{-1}(0)\big/\underset{T}{\sim}\ =\ \mu^{-1}(0)\big/\GG_\lambda
\]
and T-duality descends to an involutive contactomorphism of the moduli space
\[
  \mathcal T_*\bigl[(\mathbf x,\mathbf y),[t]\bigr]
  =\bigl[(-\mathbf y,-\mathbf x),[t+\mathbf x\Theta_g\mathbf y^T]\bigr].
\]
\item  If $\lambda\ne1$,  the rank of the
isotropy is the same at every level $\lambda$,
\[
  \operatorname{rank}G_{(\mathbf x,\mathbf y)}=4g-\operatorname{rank}\operatorname{im}(l^{0}_{(A,B)})
\]
and the level dependence is  captured by the
finite group $\ZZ^{4g}/G_{(\mathbf x,\mathbf y)}\cong\operatorname{im}(l_{(A,B)})$. At a point where
$\operatorname{Tors}\operatorname{im}(l^{0}_{(A,B)})\cong\ZZ/N$, with $N$ the smallest positive
integer for which $N\,\phi^{0}_{(\mathbf m,\mathbf n)}(\mathbf x,\mathbf y)\in\ZZ$ for all
$(\mathbf m,\mathbf n)\in\ZZ^{2g}\times\ZZ^{2g}$, we have
\[
  \operatorname{Tors}\operatorname{im}(l_{(A,B)})\ \cong\ \ZZ\big/\tfrac{N}{\gcd(N,\,\lambda-1)}.
\]
\end{enumerate}
\end{theorem}
\begin{proof}
(i) At $\lambda=1$ we have $l_{(A,B)}=(\lambda-1)\,l^{0}_{(A,B)}\equiv0$ by \eqref{eq:l-map}, so
$G_{(\mathbf x,\mathbf y)}=\ker l_{(A,B)}=\ZZ^{2g}\times\ZZ^{2g}$ for every $(A,B)$ and by Theorem~\ref{main2}(i),
 this gives that  $\GG^T_{(A,B)}=\GG_0\oplus(\ZZ^{2g}\times\ZZ^{2g})=\GG_\lambda$ is
the full gauge group. Hence $\underset{T}{\sim}$ is ordinary gauge equivalence and
$\mu^{-1}(0)/\underset{T}{\sim}=\mu^{-1}(0)/\GG_\lambda$. Moreover $\phi^{1}_{(\mathbf m,\mathbf n)}(\mathbf x,\mathbf y)
=\mathbf n\Theta_g\mathbf x^T$ is the  cocycle of Lemma~\ref{lem:finite-model} at $\lambda=1$,
so under the identification of Theorem~\ref{main2}(ii) the relation $\approx_T$ is exactly the
 relation presenting $\mu^{-1}(0)/\GG_\lambda$. Hence $\mathcal T_*$ descends to the moduli space.
Its formula is that of Theorem~\ref{main2}(iii), and it preserves the contact form
$\alpha=dt-\mathbf y\Theta_g\,d\mathbf x^T$, since
\[
  \sigma^*\alpha=d\bigl(t+\mathbf x\Theta_g\mathbf y^T\bigr)-(-\mathbf x)\Theta_g\,d(-\mathbf y)^T
  =dt+d\mathbf x\,\Theta_g\mathbf y^T+\mathbf x\Theta_g\,d\mathbf y^T-\mathbf x\Theta_g\,d\mathbf y^T
  =dt-\mathbf y\Theta_g\,d\mathbf x^T=\alpha ,
\]
using $d\mathbf x\,\Theta_g\mathbf y^T=-\mathbf y\Theta_g\,d\mathbf x^T$, so $\mathcal T_*$ is an
involutive contactomorphism of the moduli space.

(ii) Let $k:=\lambda-1\ne0$. Since $l_{(A,B)}=k\cdot l^{0}_{(A,B)}$, its image is
$\operatorname{im}(l_{(A,B)})=k\cdot\operatorname{im}(l^{0}_{(A,B)})$, which gives the first
assertion. For the rank, write $S:=\operatorname{im}(l^{0}_{(A,B)})\le\RR/\ZZ$, a finitely
generated subgroup. Multiplication by $k$ on $\RR/\ZZ$ has kernel $\tfrac1k\ZZ/\ZZ$, a finite
group of order $|k|$, so $S\cap\ker(k\cdot)$ is finite and
\[
  \operatorname{rank}\bigl(\operatorname{im}(l_{(A,B)})\bigr)
  =\operatorname{rank}(kS)
  =\operatorname{rank}(S)-\operatorname{rank}\bigl(S\cap\ker(k\cdot)\bigr)
  =\operatorname{rank}(S).
\]
Thus $\operatorname{rank}\operatorname{im}(l_{(A,B)})=\operatorname{rank}\operatorname{im}(l^{0}_{(A,B)})$,
independently of $\lambda\ne1$. Since $l_{(A,B)}\colon\ZZ^{2g}\times\ZZ^{2g}\to\RR/\ZZ$ has
$G_{(\mathbf x,\mathbf y)}=\ker l_{(A,B)}$, this kernel is free. Hence $\ZZ^{4g}/G_{(\mathbf x,\mathbf y)}\cong\operatorname{im}(l_{(A,B)})$ and the rank–nullity identity
over $\ZZ$  yields
\[
  \operatorname{rank}G_{(\mathbf x,\mathbf y)}
  =4g-\operatorname{rank}\operatorname{im}(l_{(A,B)})
  =4g-\operatorname{rank}\operatorname{im}(l^{0}_{(A,B)}).
\]

For the finite part, decompose $S=S_{\mathrm{free}}\oplus\operatorname{Tors}S$ with
$\operatorname{Tors}S\cong\ZZ/N$ finite cyclic, a finite subgroup of $\RR/\ZZ$ being cyclic.
Here $N$ is the order of $\operatorname{Tors}S$,  that is the smallest positive integer with
$N\,\phi^{0}_{(\mathbf m,\mathbf n)}(\mathbf x,\mathbf y)\in\ZZ$ for all $(\mathbf m,\mathbf n)\in\ZZ^{2g}\times\ZZ^{2g}$.
Multiplication by $k$ respects this decomposition; on $S_{\mathrm{free}}$ it is injective, so
$k\,S_{\mathrm{free}}$ is again free of the same rank and contributes no torsion, while on
$\ZZ/N$ its image is the subgroup generated by $k$, namely $\ZZ/(N/\gcd(N,k))$. Hence
$\operatorname{Tors}\operatorname{im}(l_{(A,B)})\cong\ZZ/(N/\gcd(N,\lambda-1))$.
\end{proof}

\subsection{The action groupoid and the isotropy cocycle}\label{Groupoid}
In this section we give a groupoid presentation of the T-dualizable moduli stack and
describe the induced action of T-duality on it, realised as an automorphism of the
associated groupoid. Throughout, $\lambda\in\ZZ\setminus\{0\}$ is the level of
the action \eqref{gautran} and, as in Remark~\ref{rem:matched-levels}, both sides of the
T-duality correspondence are twisted at the same level.

A feature of the twisted theory, invisible at level $0$, is that two distinct
$\RR/\ZZ$-valued cocycles are in play; the cocycle $\phi^{\lambda}$ of
Lemma~\ref{lem:finite-model}, which governs the deck action on the fibre and hence the topology
of the moduli space, and the cocycle $(\lambda-1)\phi^{0}$ measuring the failure of T-duality to be equivariant, with $\phi^{0}$ as in
\eqref{eq:phi-zero}. 

The moduli space of flat mixed fields is presented by the action Lie groupoid $\mathcal H$ with object space
\[
  \mathcal H^{(0)} \;=\; (\mathbb{R}^{2g}\times\mathbb{R}^{2g})\times \mathbb{R}/\mathbb{Z}
\]
and morphism space
\[
  \mathcal H^{(1)} \;=\; (\mathbb{R}^{2g}\times\mathbb{R}^{2g})\times(\mathbb{Z}^{2g}\times\mathbb{Z}^{2g})\times \mathbb{R}/\mathbb{Z},
\]
with structure maps
\[
  s\big((\mathbf{x},\mathbf{y}),(\mathbf{m},\mathbf{n}),[t]\big)=\big((\mathbf{x},\mathbf{y}),[t]\big),\qquad
  t\big((\mathbf{x},\mathbf{y}),(\mathbf{m},\mathbf{n}),[t]\big)=\big((\mathbf{x}+\mathbf{m},\mathbf{y}+\mathbf{n}),[t+\phi^{\lambda}_{(\mathbf{m},\mathbf{n})}(\mathbf{x},\mathbf{y})]\big),
\]
where, by Lemma~\ref{lem:finite-model},
\[
  \phi^{\lambda}_{(\mathbf{m},\mathbf{n})}(\mathbf{x},\mathbf{y}) \;=\; \mathbf{n}\,\Theta_g\,\mathbf{x}^{T}+(\lambda-1)\,\mathbf{m}\,\Theta_g\,\mathbf{y}^{T}.
\]
The action is free and proper,
and the orbit space is the moduli space
\[
  \mathcal H^{(0)}/\mathcal H \;=\;  \mu^{-1}(0)/\GG_\lambda \;=\; M_{|\lambda|}.
\]
 
The T-dualizability condition of Proposition~\ref{prop:obstruction}, namely
$(\lambda-1)\,\phi^{0}_{(\mathbf{m},\mathbf{n})}(\mathbf{x},\mathbf{y})\in\mathbb{Z}$, is encoded
by the continuous groupoid cocycle
\[
  \Phi\colon \mathcal H \longrightarrow \mathbb{R}/\mathbb{Z},\qquad
  \Phi\big((\mathbf{x},\mathbf{y},[t]),(\mathbf{m},\mathbf{n})\big)=\big[(\lambda-1)\,\phi^{0}_{(\mathbf{m},\mathbf{n})}(\mathbf{x},\mathbf{y})\big].
\]
A direct computation using
$\phi^{0}_{(\mathbf{m},\mathbf{n})}(\mathbf{x}+\mathbf{m}',\mathbf{y}+\mathbf{n}')=\phi^{0}_{(\mathbf{m},\mathbf{n})}(\mathbf{x},\mathbf{y})+\big(\mathbf{n}\Theta_g \mathbf{m}'^{T}-\mathbf{m}\Theta_g \mathbf{n}'^{T}\big)$
and the integrality of $\Theta_g$ shows that $\Phi$ is multiplicative,
$$\Phi(\gamma_1\gamma_2)=\Phi(\gamma_1)+\Phi(\gamma_2)$$ in $\mathbb{R}/\mathbb{Z}$,
so $\Phi$ is a well-defined $\mathbb{R}/\mathbb{Z}$-valued groupoid cocycle. Note that $\Phi$ is
identically zero precisely when $\lambda=1$.
 
 \begin{definition}
The \emph{T-dualizable groupoid} is the kernel subgroupoid of $\mathcal H$, 
\[
  \mathcal K \;=\; \ker\Phi
  \;=\;\Big\{\big((\mathbf{x},\mathbf{y},[t]),(\mathbf{m},\mathbf{n})\big)\;\big|\;  (\lambda-1)\phi^{0}_{(\mathbf{m},\mathbf{n})}(\mathbf{x},\mathbf{y})\in\mathbb{Z}\Big\}
  \;\rightrightarrows\; \mathbb{R}^{2g}\times\mathbb{R}^{2g}\times\mathbb{R}/\mathbb{Z}.
\]
 \end{definition}
The orbit space of $\mathcal K$ is the T-dualizable moduli stack of flat mixed fields, and the inclusion
$\mathcal K \hookrightarrow \mathcal H$ induces a canonical surjection from the stack to the
moduli space.  

To isolate the base structure of $\mathcal K$, it is convenient to work over the torus
$T^{2g}\times T^{2g}$. Writing $(\bar{\mathbf{x}},\bar{\mathbf{y}})\in T^{2g}\times T^{2g}$ for
the class of $(\mathbf{x},\mathbf{y})$, the defining condition
$(\lambda-1)\phi^{0}_{(\mathbf{m},\mathbf{n})}(\mathbf{x},\mathbf{y})\in\mathbb{Z}$ descends to
the torus. For another lift $(\mathbf{x}+\mathbf{p},\mathbf{y}+\mathbf{q})$ with
$(\mathbf{p},\mathbf{q})\in\mathbb{Z}^{2g}\times\mathbb{Z}^{2g}$,
\[
  (\lambda-1)\phi^{0}_{(\mathbf{m},\mathbf{n})}(\mathbf{x}+\mathbf{p},\,\mathbf{y}+\mathbf{q})
    = (\lambda-1)\phi^{0}_{(\mathbf{m},\mathbf{n})}(\mathbf{x},\mathbf{y})
      + (\lambda-1)\bigl(\mathbf{n}\Theta_g \mathbf{p}^T - \mathbf{m}\Theta_g \mathbf{q}^T\bigr),
\]
and the correction term lies in $\mathbb{Z}$ by integrality of $\Theta_g$, so membership in
$\mathcal K$ is independent of the chosen lift, and the isotropy group
\[
  G_{(\mathbf{x},\mathbf{y})}
    = \big\{(\mathbf{m},\mathbf{n})\in\mathbb{Z}^{2g}\times\mathbb{Z}^{2g} \ \big|\
        (\lambda-1)\phi^{0}_{(\mathbf{m},\mathbf{n})}(\mathbf{x},\mathbf{y})\in\mathbb{Z}\big\}
\]
depends only on the class $(\bar{\mathbf{x}},\bar{\mathbf{y}})$. Since
$(\mathbf{m},\mathbf{n})\in\mathbb{Z}^{2g}\times\mathbb{Z}^{2g}$ acts trivially on the base,
$\overline{\mathbf{x}+\mathbf{m}} = \bar{\mathbf{x}}$ and
$\overline{\mathbf{y}+\mathbf{n}} = \bar{\mathbf{y}}$, every morphism of $\mathcal K$ is a loop at
its base point $(\bar{\mathbf{x}},\bar{\mathbf{y}})$, carrying the fibre translation by
$\phi^{\lambda}_{(\mathbf{m},\mathbf{n})}(\mathbf{x},\mathbf{y})$. The  isotropy bundle of $\mathcal K$ over $T^{2g}\times T^{2g}$ is the bundle of abelian groups
\[
\mathcal{\widetilde{K}}^{(1)}
    = \coprod_{(\bar{\mathbf{x}},\bar{\mathbf{y}})\in T^{2g}\times T^{2g}}
      \{(\bar{\mathbf{x}},\bar{\mathbf{y}})\}\times G_{(\mathbf{x},\mathbf{y})}
\]
whose fibre over $(\bar{\mathbf{x}},\bar{\mathbf{y}})$ is the isotropy group
 $G_{(\mathbf{x},\mathbf{y})}=\ker l_{(A,B)}$, which varies with the arithmetic of
the base point. By Theorem~\ref{thm:level-stack}(ii) its rank is
$4g-\operatorname{rank}\bigl(\im l^{0}_{(A,B)}\bigr)$, independent of the level $\lambda\ne1$,
so the stratification described below is the same at every such level and only the finite
parts of the fibres depend on $\lambda$. 
The rank of $G_{(\mathbf{x},\mathbf{y})}$ is controlled by the number of rationally
independent coordinates of $(\bar{\mathbf{x}},\bar{\mathbf{y}})$. At a generic point, where
the coordinates are independent over $\mathbb{Q}$, the image $\im l_{(A,B)}$ is dense of full
rank $4g$ in $\mathbb{R}/\mathbb{Z}$, so $G_{(\mathbf{x},\mathbf{y})}$ has rank $0$ and (for
$\lambda\ne1$) is finite. At a fully rational point the image is finite cyclic and
$G_{(\mathbf{x},\mathbf{y})}$ is a finite-index sublattice of $\mathbb{Z}^{4g}$ of full rank
$4g$. Across the intermediate strata the rank drops step by step from $4g$ to $0$ as the
number of independent coordinates decreases. The fibre therefore jumps in rank along the
rational hyperplanes of the base, so $\widetilde{\mathcal K}^{(1)}\to T^{2g}\times T^{2g}$ is
not locally trivial and its total space is non-Hausdorff.

The T-duality map $\mathcal{T}_*$ on the stack is covered on $\mathcal H$ by the map
\[
  \Psi\big((\mathbf{x},\mathbf{y},[t]),(\mathbf{m},\mathbf{n})\big)
  =\Big(\big(\sigma(\mathbf{x},\mathbf{y}),[t+\mathbf{x}\Theta_g \mathbf{y}^{T}]\big),\;\tau(\mathbf{m},\mathbf{n})\Big),
\]
where $\sigma(\mathbf{x},\mathbf{y})=(-\mathbf{y},-\mathbf{x})$ and
$\tau(\mathbf{m},\mathbf{n})=(-\mathbf{n},-\mathbf{m})$, a bijection of the underlying object
and morphism sets. Its failure to intertwine the deck action is exactly the obstruction of
Proposition~\ref{prop:obstruction}: for the object map
$\psi_0(\mathbf{x},\mathbf{y},[t])=\big(\sigma(\mathbf{x},\mathbf{y}),[t+\mathbf{x}\Theta_g \mathbf{y}^{T}]\big)$,
the same computation gives
\[
  \psi_0\big((\mathbf{m},\mathbf{n})\cdot(\mathbf{x},\mathbf{y},[t])\big)
  -\tau(\mathbf{m},\mathbf{n})\cdot\psi_0\big((\mathbf{x},\mathbf{y},[t])\big)
  \;\equiv\;-(\lambda-1)\,\phi^{0}_{(\mathbf{m},\mathbf{n})}(\mathbf{x},\mathbf{y})\pmod{\mathbb{Z}}.
\]
Hence $\Psi$ is a groupoid automorphism of $\mathcal H$ precisely on the locus where this
vanishes. For $\lambda=1$ where stack and moduli space coincide, it is an automorphism of $\mathcal H$ itself and $\mathcal K=\mathcal H$,
while for $\lambda\ne1$ it is equivariant for the deck action exactly on
$G_{(\mathbf{x},\mathbf{y})}$ and so descends to an automorphism of $\mathcal K=\ker\Phi$.


\begin{thebibliography}{99}

\bibitem{A85}
M.~F.~Atiyah,
Circular symmetry and stationary-phase approximation,
Ast\'erisque, {\bf 131} (1985), 43--59.

\bibitem{Alvarez2}
E.~Alvarez, L.~Alvarez-Gaume, J.~L.~F.~Barbon and Y.~Lozano,
{\em Some Global Aspects of Duality in String Theory},
Nucl. Phys. B {\bf 415} (1994), 71--100.

\bibitem{AB}
M.~F.~Atiyah and R.~Bott,
{\em Yang-Mills equations over Riemann surfaces},
Philos. Trans. Roy. Soc. London Ser. A {\bf 308} (1983), no. 1505, 523--615.

\bibitem{B85}
J.-M.~Bismut,
Index theorem and equivariant cohomology on the loop space,
Comm. Math. Phys. {\bf 98} (1985), no. 2, 213--237.

\bibitem{BM00}
P.~Bouwknegt and V.~Mathai,
D-branes, B-fields and twisted $K$-theory,
J. High Energy Phys. {\bf 03} (2000), 007.

\bibitem{BCMMS}
P.~Bouwknegt, A.~Carey, V.~Mathai, M.~Murray and D.~Stevenson,
Twisted $K$-theory and $K$-theory of bundle gerbes,
Comm. Math. Phys. {\bf 228} (2002), 17--45.

\bibitem{BEM04a}
P.~Bouwknegt, J.~Evslin and V.~Mathai,
T-duality: topology change from $H$-flux,
Comm. Math. Phys. {\bf 249} (2004), 383--415.

\bibitem{BEM04b}
P.~Bouwknegt, J.~Evslin and V.~Mathai,
On the topology and flux of T-dual manifolds,
Phys. Rev. Lett. {\bf 92} (2004), 181601.

\bibitem{BMPR}
P.~Bouwknegt and V.~Mathai,
Review of T-duality, in progress.

\bibitem{Bry}
J.-L.~Brylinski,
{\em Loop spaces, characteristic classes and geometric quantization},
Progress in Mathematics, {\bf 107}, Birkh\"auser Boston, Inc., Boston, MA, 1993.

\bibitem{Bry98}
J.-L.~Brylinski,
{\em Categories of vector bundles and Yang-Mills equations},
Higher category theory, Contemp. Math., {\bf 230}, Amer. Math. Soc., Providence, RI, 1998, 83--98.

\bibitem{BunkeSchick}
U.~Bunke and T.~Schick,
On the topology of T-duality,
Rev. Math. Phys. {\bf 17} (2005), no. 1, 77--112.

\bibitem{Buscher}
T.~Buscher,
{\em A symmetry of the string background field equations},
Phys. Lett. B {\bf 194} (1987), no. 1, 59--62.

\bibitem{CG}
G.~R.~Cavalcanti and M.~Gualtieri,
Generalized complex geometry and T-duality,
in {\em A celebration of the mathematical legacy of Raoul Bott},
CRM Proc. Lecture Notes, {\bf 50}, Amer. Math. Soc., Providence, RI, 2010, 341--365.

\bibitem{DH82}
J.~J.~Duistermaat and G.~J.~Heckman,
On the variation in the cohomology of the symplectic form of the reduced phase space,
Invent. Math. {\bf 69} (1982), no. 2, 259--268.

\bibitem{DH83}
J.~J.~Duistermaat and G.~J.~Heckman,
Addendum to: On the variation in the cohomology of the symplectic form of the reduced phase space,
Invent. Math. {\bf 72} (1983), no. 1, 153--158.

\bibitem{GJP}
E.~Getzler, J.~D.~S.~Jones and S.~Petrack,
Differential forms on loop spaces and the cyclic bar complex,
Topology {\bf 30} (1991), no. 3, 339--371.

\bibitem{GS}
V.~Guillemin and S.~Sternberg,
{\em Symplectic techniques in physics},
Cambridge University Press, Cambridge, 1984.

\bibitem{FH08}
F.~Han,
{\em Supersymmetric QFTs, Super Loop Spaces and Bismut Chern Character},
Ph.D. thesis, University of California, Berkeley, 2008.

\bibitem{HM15}
F.~Han and V.~Mathai,
Exotic twisted equivariant cohomology of loop spaces, twisted Bismut-Chern character and T-duality,
Comm. Math. Phys. {\bf 337} (2015), no.~1, 127--150.

\bibitem{HM18}
F.~Han and V.~Mathai,
T-duality in an $H$-flux: exchange of momentum and winding,
Comm. Math. Phys. {\bf 363} (2018), no.~1, 333--350.

\bibitem{HM22}
F.~Han and V.~Mathai,
T-duality, vertical holonomy line bundles and loop Hori formulae,
Rev. Math. Phys. {\bf 33} (2022), 2250019, 25~pp.

\bibitem{HM24}
F.~Han and V.~Mathai,
T-duality with $H$-flux for 2d sigma models,
Comm. Math. Phys. {\bf 405} (2024), article no.~294.

\bibitem{HM25}
F.~Han and V.~Mathai,
Loop Hori formulae for T-duality and twisted Bismut-Chern character,
Adv. Theor. Math. Phys., to appear.

\bibitem{HLSUZ}
C.~M.~Hull, U.~Lindstrom, L.~Melo dos Santos, R.~von Unge and M.~Zabzine,
{\em Topological sigma models with $H$-flux},
J. High Energy Phys. {\bf 12} (2008), 057.

\bibitem{JP}
J.~D.~S.~Jones and S.~B.~Petrack,
The fixed point theorem in equivariant cohomology,
Trans. Amer. Math. Soc. {\bf 322} (1990), 35--49.

\bibitem{KL}
A.~Kapustin and Y.~Li,
{\em Topological sigma-models with $H$-flux and twisted generalized complex manifolds},
Adv. Theor. Math. Phys. {\bf 11} (2007), no. 2, 261--290.

\bibitem{Kleiman}
S.~L.~Kleiman,
The Picard scheme,
{\em Alexandre Grothendieck: a mathematical portrait}, International Press, Somerville, MA, 2014, 35--74.

\bibitem{Kostant}
B.~Kostant,
Quantization and unitary representations,
in {\em Lectures in modern analysis and applications III},
Lecture Notes in Math., {\bf 170}, Springer, Berlin, 1970, 87--208.

\bibitem{ALP}
H.~Aslaksen, S.~T.~Lee and J.~A.~Packer,
{\em $K$-theory for the integer Heisenberg groups},
$K$-Theory {\bf 16} (1999), no. 3, 201--227.

\bibitem{LP}
S.~T.~Lee and J.~A.~Packer,
{\em The cohomology of the integer Heisenberg groups},
J. Algebra {\bf 184} (1996), no. 1, 230--250.

\bibitem{LP99}
S.~T.~Lee and J.~A.~Packer,
{\em K-theory for $C^{*}$-algebras associated to lattices in Heisenberg Lie groups},
J. Operator Theory {\bf 41} (1999), no. 2, 291--319.

\bibitem{LRRUZ}
U.~Lindstr\"om, M.~Ro\v{c}ek, I.~Ryb, R.~von Unge and M.~Zabzine,
{\em T-duality and generalized K\"ahler geometry},
J. High Energy Phys. {\bf 02} (2008), 056.

\bibitem{MR05}
V.~Mathai and J.~Rosenberg,
T-duality for torus bundles with H-fluxes via noncommutative topology,
Comm. Math. Phys. {\bf 253} (2005), no. 3, 705--721.

\bibitem{MS}
V.~Mathai and D.~Stevenson,
Chern character in twisted $K$-theory: equivariant and holomorphic cases,
Comm. Math. Phys. {\bf 236} (2003), 161--186.

\bibitem{MW}
J.~Marsden and A.~Weinstein,
Reduction of symplectic manifolds with symmetry,
Rep. Mathematical Phys. {\bf 5} (1974), no. 1, 121--130.

\bibitem{Segal}
G.~Segal,
Topological structures in string theory,
Phil. Trans. R. Soc. Lond. A {\bf 359} (2001), 1389--1398.

\end{thebibliography}
\end{document}